\documentclass[a4paper,UKenglish,cleveref,autoref,numberwithinsect]{lipics-v2021}
\usepackage[linesnumbered, ruled, vlined]{algorithm2e}
\usepackage{tikz}
\usetikzlibrary{patterns}

\nolinenumbers

\SetAlgoLined

\newcommand{\Oh}{\mathrm{O}}
\newcommand{\oh}{\mathrm{o}}
\newcommand{\OPT}{\mathrm{OPT}}
\newcommand{\opt}{\mathrm{opt}}

\newcommand{\Ical}{\mathcal{I}}
\newcommand{\ALG}{\ensuremath{\mathrm{ALG}}}
\newcommand{\BAL}{\ensuremath{\mathrm{BAL}}}
\newcommand{\Sorting}{\textsc{Sorting}\xspace}
\newcommand{\Minimum}{\textsc{Minimum}\xspace}
\newcommand{\Selection}{\textsc{Selection}\xspace}

\newcommand{\interval}[6]{
    \draw[#5] (#2, #6) node[anchor=east]{#1} -- (#3, #6);
    \fill[#5] (#4, #6) circle (0.03cm);
}
\newcommand{\intervalsimple}[5]{
    \draw (#2, #5) node[anchor=east]{#1} -- (#3, #5);
    \fill (#4, #5) circle (0.08cm);
}

\title{Round-Competitive Algorithms for Uncertainty Problems with Parallel Queries}

\author{Thomas Erlebach}{School of Informatics, University of Leicester, UK \and \url{https://www.cs.le.ac.uk/people/te17/}}{te17@leicester.ac.uk}{https://orcid.org/0000-0002-4470-5868}{}

\author{Michael Hoffmann}{School of Informatics, University of Leicester, UK}{mh55@leicester.ac.uk}{}{}

\author{Murilo Santos de Lima\footnote{Corresponding author.}}{School of Informatics, University of Leicester, UK \and \url{https://www.ime.usp.br/~mslima/}}{mslima@ic.unicamp.br}{https://orcid.org/0000-0002-2297-811X}{}

\authorrunning{T. Erlebach, M. Hoffmann, and M.\,S. de Lima}

\Copyright{Thomas Erlebach, Michael Hoffmann, and Murilo S. de Lima}

\ccsdesc[500]{Theory of Computation~Design and analysis of algorithms}
\ccsdesc[200]{Mathematics of computing~Discrete mathematics}
\ccsdesc[100]{Theory of computation~Theory and algorithms for application domains}

\keywords{online algorithms, competitive analysis, explorable uncertainty, parallel algorithms, minimum problem, selection problem}

\relatedversion{An extended abstract is to appear in the proceedings of the 38th International Symposium on Theoretical Aspects of Computer Science (STACS 2021).}

\funding{This research was supported by EPSRC grant EP/S033483/1.}

\acknowledgements{We would like to thank Markus Jablonka for helpful discussions.}

\hideLIPIcs  

\begin{document}

\maketitle

\begin{abstract}
The area of computing with uncertainty considers problems where some information about
the input elements is uncertain, but can be obtained using queries. For example,
instead of the weight of an element, we may be given an interval that is guaranteed
to contain the weight, and a query can be performed to reveal the weight.
While previous work has considered models where queries are
asked either sequentially (adaptive model) or all at once (non-adaptive model),
and the goal is to minimize the number of queries that are needed to solve the
given problem,
we propose and study a new model where~$k$ queries can be made in parallel
in each round, and the goal is to minimize the number of query rounds.
We use competitive analysis and present upper and lower bounds on the number of
query rounds required by any algorithm in comparison with the optimal number
of query rounds.
Given a set of uncertain elements and a family of~$m$ subsets
of that set, we present an algorithm for determining the value of the minimum of each of the
subsets that requires at most $(2+\varepsilon) \cdot \mathrm{opt}_k+\mathrm{O}\left(\frac{1}{\varepsilon} \cdot \lg m\right)$
rounds for every $0<\varepsilon<1$, where $\mathrm{opt}_k$ is the optimal number of rounds,
as well as nearly matching lower bounds.
For the problem of determining the $i$-th smallest value and identifying all elements with that value in a set of uncertain elements, we give a $2$-round-competitive algorithm.
We also show that the problem of sorting a family of sets of uncertain elements admits a $2$-round-competitive algorithm and this is the best possible.
\end{abstract}

\section{Introduction}
Motivated by real-world applications where only rough information about the
input data is initially available but precise information can be obtained at
a cost, researchers have considered a range of {\bf uncertainty problems
with  queries}~\cite{bruce05uncertainty,erlebach16cheapestset,erlebach08steiner_uncertainty,feder07pathsqueires,feder03medianqueries,goerigk15knapsackqueries,megow17mst}. This research area has also been
referred to as {\bf queryable uncertainty}~\cite{erlebach15querysurvey} or {\bf explorable uncertainty}~\cite{focke17mstexp}.
For example, in the input to a sorting problem, we may be given for each
input element, instead of its precise value, only an interval containing that point.
Querying an element reveals its precise value. The goal is
to make as few queries as possible until enough information has been obtained
to solve the sorting problem, i.e., to determine a linear order of the
input elements that is consistent with the linear order of the precise values.
Motivation for explorable uncertainty comes from many different areas (see
\cite{erlebach15querysurvey} and the references given there for further examples):
The uncertain input elements may, e.g., be locations of mobile nodes or
approximate statistics derived from a distributed database cache~\cite{olston2000queries}.
Exact information can be obtained at a cost, e.g., by requesting GPS
coordinates from a mobile node, by querying the master database or by a distributed consensus algorithm.

The main model that has been studied in the explorable uncertainty setting
is the {\bf adaptive query model}: The algorithm makes queries one by one, and
the results of previous queries can be taken into account when determining
the next query. The number of queries made by the algorithm is then compared
with the best possible number of queries for the given input (i.e., the minimum
number of queries sufficient to solve the problem) using competitive
analysis~\cite{borodin98online_alg}. An algorithm is {\bf $\rho$-query-competitive} (or simply $\rho$-competitive)
if it makes at most~$\rho$ times as many queries as an optimal query set.
A very successful algorithm design paradigm in this area is based
on the concept of {\bf witness sets}~\cite{bruce05uncertainty,erlebach08steiner_uncertainty}.
A witness set is a set of input elements
for which it is guaranteed that every query set that solves the problem
contains at least one query in that set. If a problem admits witness sets of size at most~$\rho$,
one obtains a $\rho$-query-competitive algorithm by repeatedly finding
a witness set and querying all its elements.

Some work has also considered the {\bf non-adaptive query model}~(see, e.g., \cite{feder07pathsqueires, merino19matroids, olston2000queries}), where
all queries are made simultaneously and the set of queries must be chosen
in such a way that they certainly reveal sufficient information
to solve the problem. In the non-adaptive query model, one is interested
in complexity results and approximation algorithms.

In settings where the execution of a query
takes a non-negligible amount of time and there are sufficient resources
to execute a bounded number of queries simultaneously, the query process
can be completed faster if queries are not executed one at a time, but
in {\bf rounds} with~$k$ simultaneous queries.
Such scenarios include e.g.~IoT environments (such as drones measuring geographic data), or teams of interviewers doing market research.
Apart from being well motivated from an application point of view,
this variation of the model is also theoretically interesting because it
poses new challenges in selecting a useful set of~$k$ queries to be made simultaneously.
Somewhat surprisingly, however, this has not been studied yet.
In this paper, we address this gap and analyze for the first time a model where
the algorithm can make up to~$k$ queries per round, for a given value~$k$. The query
results from previous rounds can be taken into account when determining the queries
to be made in the next round.
Instead of minimizing the total number of queries, we are interested in
minimizing the number of query rounds, and we say that an algorithm
is {\bf $\rho$-round-competitive} if, for any input, it requires at most~$\rho$ times
as many rounds as the optimal query set.

A main challenge in the setting with $k$ queries per round is that
the witness set paradigm alone is no longer sufficient for obtaining a good algorithm.
For example, if a problem admits witness sets with at most $2$ elements, this immediately
implies a $2$-query-competitive algorithm for the adaptive model, but
only a $k$-round-competitive algorithm for the model with $k$ queries per round.
(The algorithm is obtained by simply querying one witness set in each round, and not making use of the
other $k-2$ available queries.) The issue is that, even
if one can find a witness set of size at most~$\rho$, the identity of subsequent witness sets
may depend on the outcome of the queries for the first witness set, and
hence we may not know how to compute a number of different witness sets
that can fill a query round if $k \gg \rho$.

\subparagraph*{Our contribution.}
Apart from introducing the model of explorable uncertainty with~$k$ queries per round,
we study several problems in this model: \Minimum, \Selection and \Sorting. For \Minimum
(or \Sorting), we assume
that the input can be a family $\mathcal{S}$ of subsets of a given ground set~$\mathcal{I}$ of uncertain elements,
and that we want to determine the value of the minimum of (or sort) all those subsets. For \Selection, we are
given a set $\mathcal{I}$ of $n$ uncertain elements and an index $i \in \{1, \ldots, n\}$, and we
want to determine the $i$-th smallest value of the $n$ precise values, and all the elements
of $\mathcal{I}$ whose value is equal to that value.

Our main contribution lies in our results for the \Minimum problem. We present an
algorithm that requires at most $(2+\varepsilon) \cdot \opt_k+\Oh\left(\frac{1}{\varepsilon} \cdot \lg m\right)$ rounds, for every
$0<\varepsilon<1$,
where~$\opt_k$ is the optimal number of rounds and~$m = |\mathcal{S}|$.
(The execution of the algorithm does not depend on~$\varepsilon$, so the upper bound holds in particular for the best choice of~$0 < \varepsilon < 1$ for given $\opt_k$ and~$m$.)
Interestingly, our algorithm follows a non-obvious approach that is reminiscent of primal-dual
algorithms, but no linear programming formulation features in the analysis.
For the case that the sets in $\mathcal{S}$ are disjoint, we obtain some improved bounds
using a more straightforward algorithm. We also
give lower bounds that apply even to the case of disjoint sets, and show that our upper bounds
are close to best possible. Note that the \Minimum problem is equivalent to the problem
of determining the maximum element of each of the sets in $\mathcal{S}$, e.g., by simply negating
all the numbers involved. A motivation for studying the \Minimum problem thus arises from the
minimum spanning tree problem with uncertain edge weights~\cite{erlebach14mstverification, erlebach08steiner_uncertainty, focke17mstexp, megow17mst}:
Determining the maximum-weight edge of each cycle of a given graph allows
one to determine a minimum spanning tree. Therefore, there is a connection between the problem of
determining the maximum of each set in a family of possibly overlapping sets (which could be
the edge sets of the cycles of a given graph) and the
minimum spanning tree problem.
The minimum spanning tree problem with uncertain edge weights has not been studied yet for the
model with $k$ queries per round, and seems to be difficult for that setting.
In particular, it is not clear in advance for which cycles of the graph a maximum-weight edge actually needs to be determined, and this makes it very difficult to determine a set of $k$ queries that are useful to be asked in parallel.
We hope that our results for \Minimum provide a first step towards solving the minimum spanning tree problem.

Another motivation for solving multiple possibly overlapping sets comes from
distributed database caches~\cite{olston2000queries},
where one wants to answer database queries using cached local data
and a minimum number of queries to the master database.
Values in the local database cache may be uncertain, and exact values can
be obtained by communicating with the central master database.
Different database queries might ask for the record with
minimum value in the field with uncertain information
among a set of database records satisfying certain criteria,
or for a list of such database records sorted by the field
with uncertain information. Answering such database queries
while making a minimum number of queries for exact values
to the master database corresponds to the \Minimum and \Sorting
problems we consider.

For the \Selection problem, we obtain a $2$-round-competitive algorithm.
For \Sorting,
we show that there is a $2$-round-competitive algorithm, by adapting ideas from
a recent algorithm for sorting in the standard adaptive model~\cite{halldorsson19sortingqueries}, and that this is best possible.

We also discuss the relationship between our model and another model of parallel queries proposed by Mei{\ss}ner~\cite{meissner18querythesis}, and we give general reductions between both settings.

\subparagraph*{Literature overview.}
The seminal paper on minimizing the number of queries to solve a problem on uncertainty intervals is by Kahan~\cite{kahan91queries}.
Given~$n$ elements in uncertainty intervals, he presented optimal deterministic adaptive algorithms for finding the maximum, the median, the closest pair, and for sorting.
Olston and Widom \cite{olston2000queries} proposed a distributed database system which exploits uncertainty intervals to improve performance.
They gave non-adaptive algorithms for finding the maximum, the sum, the average and for counting problems.
They also considered the case in which errors are allowed within a given bound, so a trade-off between performance and accuracy can be achieved.
Khanna and Tan \cite{khanna01queries} extended this previous work by investigating adaptive algorithms for the situation in which bounded errors are allowed.
They also considered the case in which query costs may be non-uniform, and presented results for the selection, sum and average problems, and for compositions of such functions.
Feder {\em et al.} \cite{feder03medianqueries} studied the generalized median/selection problem,
presenting optimal adaptive and non-adaptive algorithms.
They proved that those are the best possible adaptive and non-adaptive algorithms, respectively, instead of evaluating them from a competitive analysis perspective.
They also investigated the {\bf price of obliviousness}, which is the ratio between the non-adaptive and adaptive strategies.

After this initial foundation, many classic discrete problems were studied in this framework, including geometric problems~\cite{bruce05uncertainty, charalambous13uncertainty}, shortest paths~\cite{feder07pathsqueires}, network verification~\cite{beerliova06netdiscovery}, minimum spanning tree \cite{erlebach14mstverification, erlebach08steiner_uncertainty, focke17mstexp, megow17mst}, cheapest set and minimum matroid base~\cite{erlebach16cheapestset, merino19matroids}, linear programming~\cite{yamaguchi18ipqueries, ryzhov12lpqueries}, traveling salesman~\cite{welz14thesisqueries}, knapsack~\cite{goerigk15knapsackqueries}, and scheduling~\cite{albers20scheduling, arantes18schedulingqueries, durr2020scheduling}.
The concept of witness sets was proposed by Bruce~{\em et~al.}~\cite{bruce05uncertainty}, and identified as a pattern in many algorithms by Erlebach and Hoffmann~\cite{erlebach15querysurvey}.
Gupta {\em et al.} \cite{gupta16queryselection} extended this framework to the setting where a query may return a refined interval, instead of the exact value of the element.

The problem of sorting uncertainty data has received some attention recently.
Halldórsson and de Lima~\cite{halldorsson19sortingqueries} presented better query-competitive algorithms, by using randomization or assumptions on the underlying graph structure.
Other related work on sorting has considered sorting with noisy information~\cite{ajtai16sortingnoise,braverman09sortingnoisy} or preprocessing the uncertain intervals so that the actual numbers can be sorted efficiently once their precise value are revealed~\cite{vanderhoog19ambiguouspoints}.

The idea of performing multiple queries in parallel was also investigated by Mei{\ss}ner~\cite{meissner18querythesis}.
Her model is different, however.
Each round/batch can query an unlimited number of intervals, but at most a fixed number of rounds can be performed.
The goal is to minimize the total number of queries.
Mei{\ss}ner gave results for selection, sorting and minimum spanning tree problems.
We discuss this model in Section~\ref{sec:meissner}.
A similar model was also studied by Canonne and Gur for property testing~\cite{cannone2018property}.

\subparagraph*{Organization of the paper.}
We present some definitions and preliminary results in Section~\ref{sec:def}.
Sections~\ref{sec:sorting},~\ref{sec:minimum} and~\ref{sec:median} are devoted to the sorting, minimum and selection problems, respectively.
In Section~\ref{sec:meissner}, we discuss the relationship between the model we study and the model of Mei{\ss}ner for parallel queries~\cite{meissner18querythesis}.
We conclude in Section~\ref{sec:future}.

\section{Preliminaries and Definitions}
\label{sec:def}
For the problems we consider, the input consists of a set of~$n$ continuous uncertainty intervals $\Ical = \{I_1, \ldots, I_n\}$ in the real line.
The precise value of each data item is $v_i \in I_i$, which can be learnt by performing a query; formally, a query on~$I_i$ replaces this interval with~$\{v_i\}$.
We wish to solve the given problem by performing the minimum number of queries (or query rounds).
We say that a closed interval $I_i = [\ell_i, u_i]$ is {\bf trivial} if $\ell_i = u_i$; clearly $I_i = \{v_i\}$, so trivial intervals never need to be queried.
Some problems require that intervals are either open or trivial; we will discuss this in further detail when addressing each problem.
For a given realization $v_1, \ldots, v_n$ of the precise values, a set~$Q \subseteq \mathcal{I}$ of intervals is a {\bf feasible query set} if querying~$Q$ is enough to solve the given problem (i.e., to output a solution that can be proved correct based only on the given intervals and the answers to the queries in $Q$), and an {\bf optimal query set} is a feasible query set of minimum size.
Since the precise values are initially unknown to the algorithm and can be defined adversarially, we have an online exploration problem~\cite{borodin98online_alg}.
We fix an optimal query set~$\OPT_1$, and we write $\opt_1 := |\OPT_1|$.
An algorithm which performs up to $\rho \cdot \opt_1$ queries is said to be {\bf $\rho$-query-competitive}.
Throughout this paper, we only consider deterministic algorithms.

In previous work on the adaptive model, it is assumed that queries are  made sequentially,
and the algorithm can take the results of all previous queries into account when deciding
the next query.
We consider a model where queries are made in {\bf rounds} and we can perform up to~$k$ queries in parallel in each round.
The algorithm can take into account the results from
all queries made in previous rounds when deciding which queries to make in the next
round. The adaptive model with sequential queries is the special case of our model
with~$k=1$.
We denote by~$\opt_k$ the optimal number of rounds to solve the given instance.
Note that $\opt_k = \lceil \opt_1 / k \rceil$ as
$\OPT_1$ only depends on the input intervals and their precise values and can be distributed into rounds of $k$ queries arbitrarily.
For an algorithm $\ALG$ we denote by $\ALG_1$ the number of queries it makes, and
by $\ALG_k$ the number of rounds it uses.
An algorithm which solves the problem in up to $\rho \cdot \opt_k$ rounds is said to be {\bf $\rho$-round-competitive}.
A query performed by an algorithm that is not in~$\OPT_1$ is called a {\bf wasted} query, and we say that the algorithm {\bf wastes} that query;
a query performed by an algorithm that is not wasted is {\bf useful}.

\begin{proposition}
\label{prop:wasted}%
If an algorithm makes all queries in $\OPT_1$, wastes $w$ queries in total over all rounds excluding the final round,
always makes $k$ queries per round except possibly in the final round, and stops as soon as the queries made so far suffice to solve the problem, then its number of rounds will be
$\left\lceil (\opt_1 + w)/k \right\rceil \leq \opt_k + \lceil w / k \rceil$.
\end{proposition}

The problems we consider are \Minimum, \Sorting and \Selection. For \Minimum and \Sorting,
we assume that we are given a set $\mathcal{I}$ of $n$ intervals and a family $\mathcal{S}$ of $m$
subsets of $\mathcal{I}$. 
For \Sorting, the task is to output, for each set $S\in\mathcal{S}$,
an ordering of the elements in $S$ that is consistent with the order of their precise values.
For \Minimum, the task is to output, for each $S\in\mathcal{S}$,
an element whose precise value is the minimum of the precise values of all elements in $S$,
along with the value of that element.\footnote{In some of the literature, it is only required to identify the element with minimum value. Returning the precise minimum value, however, is also an important problem, as discussed in \cite[Section~7]{megow17mst} for the minimum spanning tree problem.}
Regarding the family~$\mathcal{S}$, we can distinguish the cases where $\mathcal{S}$ contains
a single set, where all sets in~$\mathcal{S}$ are pairwise disjoint, and the case
where the sets in $\mathcal{S}$ may overlap, i.e., may have common elements.
For \Selection, we are given a set $\mathcal{I}$ of $n$
intervals and an index~$i \in \{1, \ldots, n\}$. The task is to output the $i$-th smallest
value $v^*$ (i.e., the value in position $i$ in a sorted list of the precise values of
the $n$ intervals), as well as the set of intervals whose precise value equals~$v^*$.
We also discuss briefly a variant of \Minimum in which we seek all elements whose precise value is the minimum and a variant of \Selection in which we only seek the value~$v^*$.

For a better understanding of the problems, we give a simple example for \Sorting with $k = 1$.
We have a single set with two intersecting intervals.
There are four different configurations of the realizations of the precise values, which are shown in Figure~\ref{fig:sorting}.
In Figure~\ref{fig:sorting1}, it is enough to query~$I_1$ to learn that $v_1 < v_2$; however, if an algorithm first queries~$I_2$, it cannot decide the order, so it must query~$I_1$ as well.
In Figure~\ref{fig:sorting2} we have a symmetric situation.
In Figure~\ref{fig:sorting3}, both intervals must be queried (i.e., the only feasible query set is $\{I_1,I_2\}$),
otherwise it is not possible to decide the order.
Finally, in Figure~\ref{fig:sorting4} it is enough to query either~$I_1$ or~$I_2$; hence, both $\{I_1\}$ and $\{I_2\}$ are feasible query sets.
Since those realizations are initially identical to the algorithm, this example shows that no deterministic algorithm can be better than $2$-query-competitive, and this example can be generalized by taking multiple copies of the given structure.
For \Minimum, however, an optimum solution can always be obtained by first querying~$I_1$ (and then $I_2$ only if necessary): Since we need the precise value of the minimum element, in Figure~\ref{fig:sorting2} it is not enough to just query~$I_2$.

\begin{figure}[bt]
  \begin{subfigure}[t]{0.25\textwidth}
   \centering
   \begin{tikzpicture}[thick, scale=0.95]
    \intervalsimple{$I_1$}{0}{2}{0.5}{0}
    \intervalsimple{$I_2$}{1}{3}{1.5}{0.5}
   \end{tikzpicture}
   \caption{}
   \label{fig:sorting1}
  \end{subfigure}\hfill
  \begin{subfigure}[t]{0.25\textwidth}
   \centering
   \begin{tikzpicture}[thick, scale=0.95]
    \intervalsimple{$I_1$}{0}{2}{1.5}{0}
    \intervalsimple{$I_2$}{1}{3}{2.5}{0.5}
   \end{tikzpicture}
   \caption{}
   \label{fig:sorting2}
  \end{subfigure}\hfill
  \begin{subfigure}[t]{0.25\textwidth}
   \centering
   \begin{tikzpicture}[thick, scale=0.95]
    \intervalsimple{$I_1$}{0}{2}{1.6}{0}
    \intervalsimple{$I_2$}{1}{3}{1.4}{0.5}
   \end{tikzpicture}
   \caption{}
   \label{fig:sorting3}
  \end{subfigure}\hfill
  \begin{subfigure}[t]{0.25\textwidth}
   \centering
   \begin{tikzpicture}[thick, scale=0.95]
    \intervalsimple{$I_1$}{0}{2}{0.5}{0}
    \intervalsimple{$I_2$}{1}{3}{2.5}{0.5}
   \end{tikzpicture}
   \caption{}
   \label{fig:sorting4}
  \end{subfigure}
  \caption{Example of \Sorting for two intervals and the possible realizations of the precise values.
  We have that $\opt_1 = 1$ in (\subref{fig:sorting1}), (\subref{fig:sorting2}) and (\subref{fig:sorting4}), and $\opt_1 = 2$ in (\subref{fig:sorting3}).
  }
  \label{fig:sorting}
\end{figure}
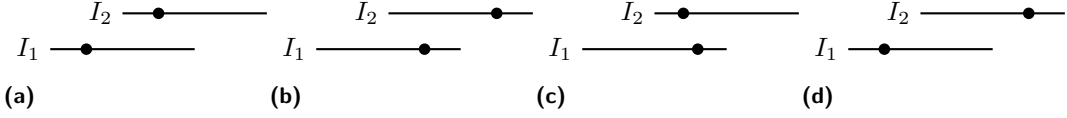

\section{Sorting}
\label{sec:sorting}

In this section we discuss the \Sorting problem.
We allow open, half-open, closed, and trivial intervals in the input, i.e., $I_i$ can be of the form $[\ell_i, u_i]$ with $\ell_i \leq u_i$, or $(\ell_i, u_i]$, $[\ell_i, u_i)$ or $(\ell_i, u_i)$ with $\ell_i < u_i$.

First, we consider the case where $\mathcal{S}$ consists of a single set~$S$, which we can assume to contain all~$n$ of the given intervals.
We wish to find a permutation $\pi : [n] \rightarrow [n]$ such that $v_i \leq v_j$ if $\pi(i) < \pi(j)$, by performing the minimum number of queries possible.
This problem was addressed for $k = 1$ in \cite{halldorsson19sortingqueries, kahan91queries, meissner18querythesis}; it admits 2-query-competitive deterministic algorithms and has a deterministic lower bound of~$2$.

For \Sorting, if two intervals $I_i = [\ell_i, u_i]$ and $I_j = [\ell_j, u_j]$ are such that $I_i \cap I_j = \{ u_i \} = \{ \ell_j \}$, then we can put them in a valid order without any further queries, because clearly $v_i \leq v_j$.
Therefore, we say that two intervals~$I_i$ and~$I_j$ {\bf intersect} (or are {\bf dependent}) if either their intersection contains more than one point, or if $I_i$ is trivial and $v_i \in (\ell_j, u_j)$ (or \emph{vice versa}).
This is equivalent to saying that $I_i$ and $I_j$ are dependent if and only if $u_i > \ell_j$ and $u_j > \ell_i$.
Two simple facts are important to notice, which are proven in~\cite{halldorsson19sortingqueries}:
\begin{itemize}
 \item For any pair of intersecting intervals, at least one of them must be queried in order to decide their relative order; i.e., any intersecting pair is a witness set.
 \item The \textbf{dependency graph} that represents this relation, with a vertex for each interval and an edge
between intersecting intervals, is an interval graph~\cite{lekkeikerker62interval}.
\end{itemize}

We adapt the $2$-query-competitive algorithm for \Sorting by Halldórsson and de Lima~\cite{halldorsson19sortingqueries} for $k = 1$ to the case of arbitrary~$k$.
Their algorithm first queries all non-trivial intervals in a minimum vertex cover in the dependency graph.
By the duality between vertex covers and independent sets, the unqueried intervals form an independent set, so no query is necessary to decide the order between them.
However, the algorithm still must query intervals in the independent set that intersect a trivial interval or the value of a queried interval.
To adapt the algorithm
to the case of arbitrary~$k$, we first compute a minimum vertex cover and fill as many rounds as necessary
with the given queries. After the answers to the queries are returned,
we use as many rounds as necessary to query the intervals of the remaining independent set
that contain a trivial point.

\begin{theorem}
\label{teo:sortingvc}
The algorithm of Halldórsson and de Lima~\cite{halldorsson19sortingqueries} yields a $2$-round-competitive algorithm for \Sorting that runs in polynomial time.
\end{theorem}

\begin{proof}
Any feasible query set is a vertex cover in the dependency graph, due to the fact that at least one interval in each intersecting pair must be queried.
Therefore a minimum vertex cover is at most the size of an optimal query set, so the first phase of the algorithm spends at most~$\opt_k$ rounds.
Since all intervals queried in the second phase are in any solution, again we spend at most another~$\opt_k$ rounds. As the minimum vertex cover problem for interval graphs can
be solved in polynomial time~\cite{gavril72chordal}, the overall algorithm is polynomial as well.
\end{proof}

The problem has a lower bound of 2 on the round-competitive factor. This can be shown by having $k c$ copies of a structure consisting of two dependent intervals, for some $c \geq 1$.
$\OPT_1$~may query only one interval in each pair, while we can force any deterministic algorithm to query both of them (cf.\ the configurations
shown in Figures~\ref{fig:sorting1} and~\ref{fig:sorting2}).
We have that $\opt_k = c$ while any deterministic algorithm will spend at least~$2c$ rounds.

We remark that the $2$-query-competitive algorithm for \Sorting with $k=1$ due to
Mei{\ss}ner~\cite{meissner18querythesis}, when adapted to the setting with arbitrary $k$
in the obvious way, only gives a bound of $2 \cdot \opt_k + 1$ rounds.
Her algorithm first greedily computes a maximal matching in the dependency graph and queries all non-trivial matched vertices, and then all
remaining intervals that contain a trivial point.

\vspace{\baselineskip}

Now we study the case of solving a
number of problems on different subsets of the same ground set of uncertain elements.
In such a setting, it may be better to perform queries that can be reused by different problems, even if the optimum solution for one problem may not query that interval.
We can reuse ideas from the algorithms for single problems that rely on the dependency graph.
We define a new dependency relation (and dependency graph) in such a way that two intervals are dependent if and only if they intersect {\em and} belong to a common set.
Note that the resulting graph may not be an interval graph, so some algorithms for single problems may not run in polynomial time for this generalization.

If we perform one query at a time ($k = 1$), then there are $2$-competitive algorithms.
One such is the algorithm by Mei{\ss}ner~\cite{meissner18querythesis} described above; since a maximal matching can be computed greedily in polynomial time for arbitrary graphs, this algorithm runs in polynomial time for non-disjoint problems.
If we can make $k \geq 2$ queries in parallel, then this algorithm performs
at most $2 \cdot \opt_k + 1$ rounds, and the analysis is tight since we may have an incomplete round in between the two phases of the algorithm.
If we relax the requirement that the algorithm runs in polynomial time, then we can obtain an algorithm that needs at most $2 \cdot \opt_k$ rounds, by first querying non-trivial intervals in a minimum vertex cover of the dependency graph (in as many rounds as necessary) and then the intervals that contain a trivial interval or the value of a queried interval (again, in as many rounds as necessary).

\section{The Minimum Problem}
\label{sec:minimum}

For the \Minimum problem, we assume without loss of generality that the intervals
are sorted by non-decreasing left endpoints; intervals with the same left endpoint can be ordered arbitrarily.
The {\bf leftmost} interval among a subset of~$\mathcal{I}$ is the one that comes earliest in this ordering.
We also assume that all
intervals are open or trivial; otherwise the problem has a trivial lower bound of~$n$ on the query-competitive ratio~\cite{gupta16queryselection}.

First, consider the case $\mathcal{S} = \{\mathcal{I}\}$, i.e., we have a single set. It is easy to see
that the optimal query set consists of all intervals whose left endpoint is strictly
smaller than the precise value of the minimum: If~$I_i$ with precise value~$v_i$
is a minimum element, then all other intervals with left endpoint strictly smaller than~$v_i$ must be queried to rule out that their value is smaller than~$v_i$, and~$I_i$ must be queried (unless it is a trivial interval) to determine the value of the minimum.
The optimal set of queries is hence a \emph{prefix} of the sorted list of uncertain
intervals (sorted by non-decreasing left endpoint).
This shows that there is a $1$-query-competitive algorithm when $k=1$,
and a $1$-round-competitive algorithm for arbitrary~$k$: In each round we simply
query the next~$k$ uncertain intervals in the order of non-decreasing left endpoint,
until the problem is solved.
For $k=1$, the same method yields a $1$-query-competitive algorithm for the case
with several sets: The algorithm can always query an interval with smallest left
endpoint for any of the sets that have not yet been solved.%
\footnote{If we want to determine all elements whose value equals the minimum, it is not hard to see that the optimal set of queries for each set is again a prefix. As all our
algorithms require only this property, we obtain corresponding results for that
problem variant, even for inputs with arbitrary closed, open and half-open intervals.}

In the remainder of
this section, we consider the case of multiple sets and $k>1$.
We first present a more general result for potentially overlapping sets, then we give better upper bounds for disjoint sets.
At the end of the section, we also present lower bounds.

Let $W(x)=x \lg x$; the inverse $W^{-1}$ of $W$ will show up in our analysis.
Note that $W^{-1}(x)=\Theta(x / \lg x)$ (see Appendix~\ref{app:W-1} for a proof).

Throughout this section, we assume w.l.o.g.\ that the optimum must make at least one query in each
set (or we consider only sets that require some query).
We also assume that any algorithm always discards from each
set all elements that are certainly not the minimum of that set, i.e., all elements
for which it is already clear based on the available information that their
value must be larger than the minimum value of the set (this is where
the right endpoints of intervals also need to be considered).
We adopt the following terminology.
A set in~$\mathcal{S}$ is {\bf solved} if we can determine the value of its minimum element.
A set is {\bf active}
at the start of a round if the queries made in previous
rounds have not solved the set yet. An active
set {\bf survives} a round if it is still active
at the start of the next round. An active set
that does not survive the current round is said to
be {\bf solved in} the current round.

To illustrate these concepts, let us discuss a first simple strategy to build a query set~$Q$ for a round.
Let~$\mathcal{P}$ be the set of intervals queried in previous rounds.
The {\bf prefix length} of an active set~$S$ is the length of the maximum prefix of elements from~$Q$ in the list of non-trivial intervals in $S \setminus \mathcal{P}$ ordered by non-decreasing left endpoints.
The algorithm proceeds by repeatedly adding to~$Q$ the leftmost non-trivial element not in~$Q \cup \mathcal{P}$ from an arbitrary active set with minimum prefix length.
We call this the {\bf balanced} algorithm, and denote it by $\BAL$.
We give an example of its execution in Figure~\ref{fig:bal}, with $m = 3$ disjoint sets and $k = 5$.
The optimum solution queries the first three elements in~$S_1$ and~$S_2$, and all elements in~$S_3$.
Since the algorithm picks an arbitrary active set with minimum prefix length, it may give preference to~$S_1$ and~$S_2$ over~$S_3$, thus wasting one query in~$S_1$ and one in~$S_2$ in round~2.
All sets are active at the beginning of round~2;
$S_1$~and~$S_2$ are solved in round~2, while~$S_3$ survives round~2.
Since~$S_1$ and~$S_2$ are solved in round~2, they are no longer active in round~3, so the algorithm no longer queries any of their elements.

\newcommand{\elementbox}[3]{
  \draw[#3] (#1, #2) rectangle (#1 + 1.2, #2 - 1);
}

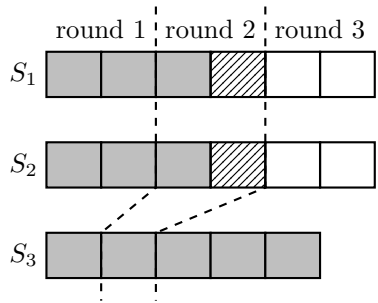
\begin{figure}[tb]
  \centering
  \begin{tikzpicture}[thick, scale=0.6]
    \elementbox{0}{0}{fill=gray!50}
    \elementbox{1.2}{0}{fill=gray!50}
    \elementbox{2.4}{0}{fill=gray!50}
    \elementbox{3.6}{0}{pattern=north east lines}
    \elementbox{4.8}{0}{}
    \elementbox{6.0}{0}{}

    \elementbox{0}{-2}{fill=gray!50}
    \elementbox{1.2}{-2}{fill=gray!50}
    \elementbox{2.4}{-2}{fill=gray!50}
    \elementbox{3.6}{-2}{pattern=north east lines}
    \elementbox{4.8}{-2}{}
    \elementbox{6.0}{-2}{}

    \elementbox{0}{-4}{fill=gray!50}
    \elementbox{1.2}{-4}{fill=gray!50}
    \elementbox{2.4}{-4}{fill=gray!50}
    \elementbox{3.6}{-4}{fill=gray!50}
    \elementbox{4.8}{-4}{fill=gray!50}
    
    \draw[dashed] (2.4, 1) -- (2.4, -3) -- (1.2, -4) -- (1.2, -5.5);
    \draw[dashed] (4.8, 1) -- (4.8, -3) -- (2.4, -4) -- (2.4, -5.5);
    
    \node at (-0.5, -0.5) {$S_1$};
    \node at (-0.5, -2.5) {$S_2$};
    \node at (-0.5, -4.5) {$S_3$};

    \node at (1.2, 0.5) {round 1};
    \node at (3.6, 0.5) {round 2};
    \node at (6.0, 0.5) {round 3};
  \end{tikzpicture}
  \caption{Possible execution of $\BAL$ for $m = 3$ disjoint sets and $k = 5$.
  Each interval is represented by a box, and the optimum solution is a prefix of each set.
  The solid boxes are useful queries, the two hatched boxes are wasted queries, and the white boxes are not queried by the algorithm.
  }
  \label{fig:bal}
\end{figure}

\subsection{The Minimum Problem with Arbitrary Sets}

We are given a set $\mathcal{I}$ of $n$ intervals and a family $\mathcal{S}$ of $m$ possibly overlapping subsets of~$\mathcal{I}$, and a number $k \geq 2$ of queries that can be performed in each round.

Unfortunately, it is possible to construct an instance in which $\BAL$ uses as many as $k \cdot \opt_k$ rounds.
Let~$c$ be a multiple of~$k$.
We have $m = c \cdot (k - 1)$ sets, which are divided in~$c$ groups with $k-1$ sets.
For $i = 1, \ldots, c$, the sets in groups $i, \ldots, c$ share the~$i$ leftmost elements.
Furthermore, each set has one extra element which is unique to that set.
The precise values are such that each set in the $i$-th group is solved after querying the first~$i$ elements.
We give an example in Figure~\ref{fig:badbal} with $k = 3$ and $c = 3$.
If we let $\BAL$ query the intervals in the order given by the indices, it is easy to see that it queries $c \cdot k$ intervals, while the~$c$ intervals that are shared by more than one set are enough to solve all sets.
In particular, note that $\BAL$ does not take into consideration that some elements are shared between different sets.
The challenge is how to balance queries between sets in a better way.

\newcommand{\bset}{
   \begin{tikzpicture}[thick, scale=2.5]
}
\newcommand{\eset}{
   \end{tikzpicture}
}
\newcommand{\ione}{\interval{$I_1$}{0}{1}{0.75}{blue}{0}}
\newcommand{\itwo}{\interval{$I_2$}{0.85}{1.85}{1.4}{black}{-0.2}}
\newcommand{\ithree}{\interval{$I_3$}{0.85}{1.85}{1.5}{black}{-0.2}}
\newcommand{\ifour}{\interval{$I_4$}{0.15}{1.15}{0.65}{red}{-0.2}}
\newcommand{\ifive}{\interval{$I_5$}{0.85}{1.85}{1.3}{black}{-0.4}}
\newcommand{\isix}{\interval{$I_6$}{0.85}{1.85}{1.55}{black}{-0.4}}
\newcommand{\iseven}{\interval{$I_7$}{0.3}{1.3}{0.55}{violet}{-0.4}}
\newcommand{\ieight}{\interval{$I_8$}{0.85}{1.85}{1.6}{black}{-0.6}}
\newcommand{\inine}{\interval{$I_9$}{0.85}{1.85}{1.7}{black}{-0.6}}

\begin{figure}[ht!]
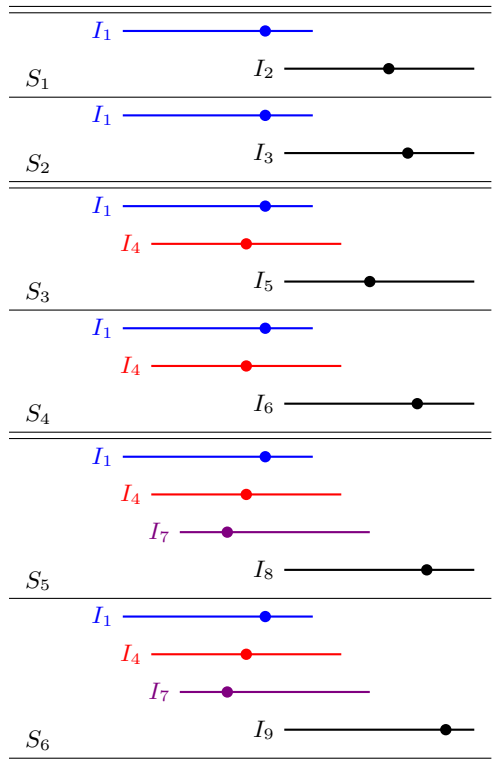

\centering
\begin{tabular}{ll}
 \hline
 \hline
 $S_1$ & \bset \ione \itwo \eset \\
 \hline
 $S_2$ & \bset \ione \ithree \eset \\
 \hline
 \hline
 $S_3$ & \bset \ione \ifour \ifive \eset \\
 \hline
 $S_4$ & \bset \ione \ifour \isix \eset \\
 \hline
 \hline
 $S_5$ & \bset \ione \ifour \iseven \ieight \eset \\
 \hline
 $S_6$ & \bset \ione \ifour \iseven \inine \eset \\
 \hline
 \hline
\end{tabular}

\caption{Bad instance for $\BAL$ with overlapping sets, with $k = 3$ and $c = 3$. $\BAL$ will query the following rounds: $\{I_1, I_2, I_3\}, \{I_4, I_5, I_6\}, \{I_7, I_8, I_9\}$. It is enough to query $\{I_1, I_4, I_7\}$.}
\label{fig:badbal}
\end{figure}

We give an algorithm that requires at most $(2+\varepsilon) \cdot \opt_k + \Oh\left(\frac{1}{\varepsilon} \cdot \lg m\right)$ rounds, for every $0 < \varepsilon < 1$.
(The execution of the algorithm does not depend on~$\varepsilon$, so the upper bound holds in particular for the best choice of~$0 < \varepsilon < 1$ for given $\opt_k$ and~$m$.)
It is inspired by how some primal-dual algorithms work.
The pseudocode for determining the queries to be made in
a round is shown in Algorithm~\ref{fig:algA}.
First, we try to include the leftmost element of each set in the set of queries~$Q$.
If those are not enough to fill a round, then we maintain a variable~$b_i$ for each set~$S_i$, which can be interpreted as a budget for each set.
The variables are increased simultaneously at the same rate, until the sets that share a current leftmost unqueried element not in~$Q$ have enough budget to buy it.
More precisely, at a given point of the execution, for each element $e \in \mathcal{I} \setminus Q$, let~$F_e$ contain the indices of the sets that have~$e$ as their leftmost unqueried element not in~$Q$.
We include~$e$ in~$Q$ when $\sum_{i \in F_e} b_i = 1$, and then we set~$b_i$ to zero for all $i \in F_e$.
We repeat this process until $|Q| = k$ or there are no unqueried elements in $\mathcal{I} \setminus Q$.

\begin{algorithm}[tb]
\KwData{family $\mathcal{S} = \{S_1,\ldots,S_m\}$ of active subsets of the ground set~$\mathcal{I}$}
\KwResult{set $Q \subseteq \mathcal{I}$ of at most $k$ queries to make}
\Begin{%
  $Q\leftarrow$ set of leftmost unqueried elements of all sets in $\mathcal{S}$\;
  \eIf{$|Q|\ge k$}{
  	$Q\leftarrow$ arbitrary subset of $Q$ with size $k$\;
  }{
  	$b_i \leftarrow 0$ for all $S_i \in \mathcal{S}$\;
	\While{$|Q|<k$ and there are unqueried elements in $\mathcal{I}\setminus Q$}{
		\ForEach{$e \in \mathcal{I} \setminus Q$}{
			$F_e \leftarrow \{i\mid$ $e$
			is the leftmost unqueried element from
			$\mathcal{I}\setminus Q$ in $S_i \}$\;
		}
		increase all $b_i$ simultaneously at the same rate
		until there is an unqueried element $e\in \mathcal{I}\setminus Q$ that
		satisfies
		$\sum_{i \in F_e} b_i= 1$\;
		$Q\leftarrow Q \cup \{e\}$\;
		$b_i \leftarrow 0$ for all $i\in F_e$\;
	}
  }
  \Return{$Q$}\;
}
\caption{Computing a query round for possibly non-disjoint sets}
\label{fig:algA}
\end{algorithm}

When a query $e$ is added to $Q$, we say that it is {\bf charged}
to the sets $S_i$ with $i\in F_e$. The amount of charge
for set $S_i$ is equal to the value of $b_i$ just before $b_i$
is reset to~$0$ after adding~$e$ to~$Q$. We also say
that the set $S_i$ {\bf pays} this amount for~$e$.

\begin{definition}
Let $\varepsilon > 0$.
A round is {\bf $\varepsilon$-good} if at least $k/2$ of the queries
made by Algorithm~\ref{fig:algA} are also in $\OPT_1$ (i.e., are useful queries),
or if at least
$a/r$ active sets are solved in that round, where $a$
is the number of active sets at the start of the round
and $r = (2(1 + \varepsilon) + \sqrt{2\varepsilon^2 +4\varepsilon + 4})/\varepsilon$.
A round that is not $\varepsilon$-good is called {\bf $\varepsilon$-bad}.
\end{definition}

Note that $r > 2$ for any $\varepsilon > 0$.

\begin{lemma}
\label{lem:badgood}%
If a round is $\varepsilon$-bad, then Algorithm~\ref{fig:algA} will make at least $2k/(2 + \varepsilon)$ useful queries
in the following round.
\end{lemma}

\begin{proof}
Let $a$ denote the number of active sets at the start
of an $\varepsilon$-bad round.
Let $s$ be the number of sets
that are solved in the
current round; note that $s < a/r$ because the current round is $\varepsilon$-bad.
Let $T$ be the total amount by which each value $b_i$
has increased during the execution of Algorithm~\ref{fig:algA}.
If the simultaneous increase of all $b_i$ is interpreted
as time passing, then $T$ corresponds to the point in time
when the computation of the set $Q$ has been completed.
For example, if some set $S_i$ did not pay for any element
during the whole execution, then $T$ is equal to the value
of $b_i$ at the end of the execution of Algorithm~\ref{fig:algA}.

Let $Q$ be the set of queries that Algorithm~\ref{fig:algA} makes
in the current round. We claim that every wasted query in $Q$
is charged only to sets that are solved in this round.
Consider a wasted query~$e$ that is in some set~$S_j$ not solved in this round.
At the time $e$ was selected, $j$~cannot have been in $F_e$ because
otherwise $e$ would be a useful query.
Therefore, we do not charge~$e$ to~$S_j$.

The total number of wasted queries is therefore
bounded by $Ts$, as these queries are paid for
by the $s$ sets solved in this round.
As the number of wasted queries in a bad round
is larger than $k/2$, we therefore have $Ts > k/2$.
As $s < a/r$, we get $k/2 < Ta/r$, so
$T> (r/2) \cdot (k/a)$.

Call a surviving set $S_i$ {\bf rich} if
$b_i> k/a$ when the computation of $Q$ is completed.
A set that is not rich is called {\bf poor}.
Note that a poor set must have spent at least
an amount of $(r/2 - 1) \cdot (k/a) > 0$, as its total budget would
be at least $T > (r/2) \cdot (k/a)$ if it had not paid for
any queries. As the poor sets have paid for
fewer than $k/2$ elements in total (as there
are fewer than $k/2$ useful queries in the current
round), the number of poor sets is bounded by
$\frac{k/2}{(r/2 - 1) \cdot (k/a)} = a / (r - 2) > 0$. As there are more than
$(1 - 1/r) \cdot a$ surviving sets and at most $a/(r-2)$ of them
are poor, there are at
least $(1 - 1 /r) \cdot a - a/(r-2) = ((r-2)(r-1) -r) / (r(r-2)) \cdot a = 2a/(2+\varepsilon) > 0$ surviving sets that are rich.

Let $e$ be any element that is the leftmost
unqueried element (at the end
of the current round) of a rich surviving set.
If $e$ was the leftmost unqueried element of
more than $a/k$ rich surviving sets, those
sets would have been able to pay for $e$
(because their total remaining budget would
be greater than $k/a \cdot a/k=1$)
before the end of the execution of Algorithm~\ref{fig:algA},
a contradiction to $e$ not being included
in~$Q$.
Hence, the number of distinct leftmost unqueried
elements of the at least $2a/(2+\varepsilon)$ rich surviving sets
is at least $(2a/(2+\varepsilon)) / (a/k) = 2k/(2+\varepsilon)$. So the following
round will query at least $2k/(2+\varepsilon)$ elements that are the leftmost unqueried
element of an active set, and all those are useful queries that are made in the next round.
\end{proof}

\begin{theorem}
Let $\opt_k$ denote the optimal number of rounds
and $A_k$ the number of rounds made if the queries
are determined using Algorithm~\ref{fig:algA}.
Then, for every $0 < \varepsilon < 1$, $A_k \le (2+\varepsilon) \cdot \opt_k + \Oh\left(\frac{1}{\varepsilon} \cdot \lg m\right)$.
\end{theorem}

\begin{proof}
In every round, one of the following must hold:
\begin{itemize}
\item The algorithm makes at least $k/2$ useful queries.
\item The algorithm solves at least a fraction of $1/r$ of the active sets.
\item If none of the above hold, the algorithm makes at least $2k/(2+\varepsilon)$ useful queries in the following round (by Lemma~\ref{lem:badgood}).
\end{itemize}
The number of rounds in which the algorithm solves at least a fraction of $1/r$ of the active
sets is bounded by $\lceil \log_{r/(r-1)} m\rceil = \Oh\left(\frac{1}{\varepsilon} \cdot \lg m\right)$, since $1/\left(\lg \frac{r}{r-1}\right) < 5/\varepsilon$ for $0 < \varepsilon < 1$. In every round where the algorithm
does not solve at least a fraction of $1/r$ of the active sets, the algorithm makes
at least $k/(2+\varepsilon)$ useful queries on average (if in any such round it makes fewer than $k/2$ useful queries,
it makes $2k/(2+\varepsilon)$ useful queries in the following round). The number of such rounds is therefore bounded
by $(2 + \varepsilon) \cdot \opt_k$.
\end{proof}

We do not know if this analysis is tight, so it would be worth investigating this question.

\subsection{The Minimum Problem with Disjoint Sets}

We now consider the case where $k \geq 2$ and the $m$~sets in the given family $\mathcal{S}$
are pairwise disjoint.
For this case, it turns out that the balanced algorithm achieves good upper bounds.

\begin{theorem}
$\BAL_k \le \opt_k + \Oh(\lg \min\{k, m\})$.
\label{th:ballog}
\end{theorem}

\begin{proof}
First we prove the bound for $m \leq k$.
Index the sets in such a way that
$S_i$ is the $i$-th set that is solved by $\BAL$, for $1 \le i \le m$.
Sets that are solved in the same round are ordered by non-decreasing
number of queries made in them in that round by $\BAL$.
In the round when~$S_i$ is solved, there are at least $m-(i-1)$ active sets,
so the number of wasted queries in~$S_i$ is at most $\frac{k}{m-(i-1)}$.
($\BAL$ makes at most $\left\lceil \frac{k}{m-(i-1)} \right\rceil$ queries in $S_i$,
and at least one of these is not wasted.)
The total number of wasted queries is then at most
 $\sum_{i=1}^{m} \frac{k}{m-(i-1)} = \sum_{i=1}^{m} k/i = k \cdot H(m)$,
where~$H(m)$ denotes the $m$-th Harmonic number.
By Proposition~\ref{prop:wasted}, $\BAL_k\le \opt_k + \Oh(\lg m)$.

If $m > k$, observe that the algorithm does not waste any queries until the number of active sets is at most~$k$.
From that point on, it wastes at most $k\cdot H(k)$ queries
following the arguments in the previous paragraph, so the number
of rounds is bounded by $\opt_k + \Oh(\log k)$.
\end{proof}

We now give a more refined analysis that
provides a better bound for $\opt_k=1$, as well as a better
multiplicative bound than what would follow from
Theorem~\ref{th:ballog}.

\begin{lemma}
If $\opt_k = 1$, then $\BAL_k \le \Oh(\lg m / \lg\lg m)$.
\label{lem:opt1}
\end{lemma}

\begin{proof}
Consider an arbitrary instance of the problem with $\opt_k = 1$. Let $R+1$ be the
number of rounds needed by the algorithm.
For each of the first $R$ rounds, we consider the fraction $b_i$
of active sets that are not solved in that round. More formally,
for the $i$-th round, for $1\le i\le R$, if $a_i$ denotes the number
of active sets at the start of round $i$ and $a_{i+1}$ the number
of active sets at the end of round $i$, then we define $b_i=a_{i+1}/a_i$.

Consider round~$i$, $1\le i\le R$.
A set that is active at the start of round~$i$ and is still active at the start of the
round~$i+1$ is called a {\bf surviving set}. A set that is active at the start
of round~$i$ and gets solved by the queries made in round~$i$ is called a {\bf solved
set}.
For each surviving set, all queries made in that set in round~$i$ are useful. For each
solved set, at least one query made in that set is useful.
We claim that this implies the algorithm makes at least~$k b_i$ useful queries
in round~$i$. To see this, observe that if the algorithm makes $\lfloor k/a_i\rfloor$
queries in a surviving set
and $\lceil k/a_i \rceil$ queries in a solved set, we can conceptually move one useful
query from the solved set to the surviving set. After this,
the~$a_{i+1}$ surviving sets contain at least $k/a_i$
useful queries on average, and hence $a_{i+1} \cdot k/a_i =b_i k$ useful queries in total.

As~$\OPT_1$ must make all useful queries and makes at most~$k$ queries in total,
we have that
$\sum_{i=1}^R kb_i \leq \opt_1 \leq k$, so $\sum_{i=1}^R b_i \le 1$.
Furthermore, as there are~$m$ active sets initially
and there is still at least one active set after round $R$, we have that
$\prod_{i=1}^R b_i = a_{R+1}/a_1 \ge 1/m$.
To get an upper bound on~$R$, we need to determine the largest possible value of $R$
for which there
exist values $b_i>0$ for $1\le i \le R$ satisfying $\sum_{i=1}^R b_i \le 1$ and $\prod_{i=1}^R b_i \ge 1/m$.
We gain nothing from choosing~$b_i$ with $\sum_{i=1}^R b_i < 1$, so we can
assume $\sum_{i=1}^R b_i =1$. In that case, the value of $\prod_{i=1}^R b_i$
is maximized if we set all~$b_i$ equal, namely $b_i= 1 / R$.
So we need to determine the largest value of $R$ that satisfies
$
\prod_{i=1}^R 1/R \ge 1/m
$,
or equivalently $R^R \le m$, or
$R \lg R \le \lg m$.
This shows that $R \le W^{-1}(\lg m)=\Oh(\lg m / \lg\lg m)$.
\end{proof}

\begin{corollary}
If $\opt_k = 1$, then $\BAL_k \le \Oh(\lg k / \lg\lg k)$.
\label{cor:W-1logk}
\end{corollary}

\begin{proof}
If $k\ge m$, then the corollary follows from Lemma~\ref{lem:opt1}.
If $k<m$, there can be at most~$k$ active sets, because the optimum performs
at most~$k$ queries since $\opt_k = 1$. Hence, we only need to consider
these~$k$ sets and can apply Lemma~\ref{lem:opt1} with $m=k$.
\end{proof}

Now we wish to extend these bounds to arbitrary~$\opt_k$.
It turns out that we can reduce the analysis for an instance with arbitrary~$\opt_k$ to the analysis for an instance with $\opt_k = 1$, assuming that $\BAL$ is implemented in a round-robin fashion.
A formal description of such an implementation is as follows:
fix an arbitrary order of the~$m$ sets of the
original problem instance as $S_1,S_2,\ldots,S_m$, and consider
it as a cyclic order where the set after $S_m$ is~$S_1$.
In each round, $\BAL$ distributes the $k$ queries to the
active sets as follows. Let $i$ be the index of the set
to which the last query was distributed in the previous
round (or let $i=m$ if we are in the first round). Then
initialize $Q=\emptyset$ and
repeat the following step~$k$ times.
Let~$j$ be the first index
after~$i$ such that~$S_j$ is active and has unqueried non-trivial
elements that are not in $Q$;
pick the leftmost unqueried non-trivial
element in $S_j\setminus Q$, insert it into $Q$, and set $i=j$.
The resulting set $Q$ is then queried.

\begin{lemma}
Assume that $\BAL$ distributes queries to active sets in a
round-robin fashion.
If $\BAL_k \le \rho$ for instances with $\opt_k = 1$, with $\rho$
independent of~$k$,
then $\BAL$ is $\rho$-round-competitive for arbitrary instances.
\label{lem:tto1}
\end{lemma}

\begin{proof}
Let $L=(\mathcal{I},\mathcal{S})$ be an instance with $\opt_k(L) = t$.
Note that $\opt_1(L) \le tk$.
Consider the instance~$L'$ which is identical to~$L$
except that the number of queries per round is $k' = tk$.
Use $\BAL'$ to refer to the solution computed by $\BAL$ for the instance~$L'$
(and also to the algorithm $\BAL$ when it is executed on instance~$L'$).
Note that $\opt_{k'}(L') = 1$ as $\opt_1(L') = \opt_1(L)$.
By our assumption, $\BAL'_{k'} \le \rho$.
We claim that this implies $\BAL_k \le \rho t$.

To establish the claim, we compare the situation when
$\BAL'$ has executed~$x$ rounds on~$L'$ with the situation when $\BAL$
has executed $x t$ rounds on~$L$. We claim that the following
two invariants hold for every~$x$:
\begin{description}
\item[(1)] The number of remaining active sets of $\BAL$ is at most that of $\BAL'$.
\item[(2)] $\BAL$ has made at least as many queries in each active set as $\BAL'$.
\end{description}
For a proof of these invariants, note that $\BAL'$ and $\BAL$ distribute
queries to sets in the same round-robin order, the only difference being
that $\BAL$ performs a round of queries whenever $k$ queries have been distributed,
while $\BAL'$ only performs a round of queries whenever $kt$ queries have
accumulated. Imagine the process by which both algorithms pick queries
as if it was executed in parallel, with both of the algorithms choosing
one query in each step.
The only case where $\BAL$ and $\BAL'$ can distribute
the next query to a different set is when $\BAL'$ distributes the next query
to a set $S_i$ that is no longer active for $\BAL$ (or all of whose non-trivial
unqueried elements have already been added by $\BAL$ to the set of queries to be performed
in the next round). This can happen
because $\BAL$ may have already made some of the queries that $\BAL'$ has distributed to
sets but not yet performed. If this happens, $\BAL$ will select for the next query
an element of a set that comes after $S_i$ in the cyclical order, so it will
move ahead of $\BAL'$ (i.e., it chooses a query now that $\BAL'$ will only
choose in a later step). Hence, at any step during this process,
$\BAL$ either picks the same next query as $\BAL'$ or is ahead of $\BAL'$.
This shows that if the invariants hold when $\BAL$ and $\BAL'$ have executed
$xt$ and $x$ rounds, respectively, then they also hold after they have
executed $(x+1)t$ and $x+1$ rounds, respectively. As the invariants clearly
hold for $x=0$, if follows that they always hold, and hence $\BAL_k\le \rho t$.
\end{proof}

Lemmas~\ref{lem:opt1} and~\ref{lem:tto1} imply the following.

\begin{corollary}
\label{cor:W-1logccomp}
$\BAL$ is $\Oh(\lg m / \lg\lg m)$-round-competitive.
\end{corollary}

Unfortunately, Corollary \ref{cor:W-1logk} cannot be
combined with Lemma~\ref{lem:tto1} directly to show that $\BAL$ is
$\Oh(\lg k / \lg\lg k)$-round-competitive, because the proof of Lemma~\ref{lem:tto1}
assumes that~$\rho$ is not a function of~$k$.
However, we can show the claim using different arguments.

\begin{lemma}
\label{lem:W-1logkcomp}
$\BAL$ is $\Oh(\lg k / \lg\lg k)$-round-competitive.
\end{lemma}

\begin{proof}
If $k \ge m$, the lemma follows from Corollary~\ref{cor:W-1logccomp}.

If $k<m$, let~$L$ be the given instance and let~$R_0$ be the number of rounds the algorithm needs
until the number of active sets falls below $k+1$ for the first time.
As the algorithm makes at most one query in each active set in the
first~$R_0$ rounds, all queries made in the first~$R_0$ rounds are
useful. Let~$L'$ be the instance at the end of round~$R_0$.
As~$L'$ has at most~$k$ active sets, $\BAL$ is
$\Oh(\lg k / \lg\lg k)$-round-competitive on~$L'$ by Corollary~\ref{cor:W-1logccomp},
and it needs at most $\Oh(\lg k / \lg\lg k) \cdot \opt_k(L')
= \Oh(\lg k / \lg\lg k) \cdot \lceil \opt_1(L')/k \rceil$ rounds
to solve~$L'$.

We have that $\opt_1(L)=k \cdot R_0 + \opt_1(L')$, and hence $\opt_k(L)=R_0 + \lceil \opt_1(L')/k \rceil$.
Thus,
\begin{eqnarray*}
\BAL_k(L) & \le &  R_0 + \Oh(\lg k / \lg\lg k) \cdot \lceil \opt_1(L')/k \rceil\\
& \le & \Oh(\lg k / \lg\lg k) \cdot (R_0 + \lceil \opt_1(L')/k \rceil)\\
& = & \Oh(\lg k / \lg\lg k) \cdot \opt_k(L),
\end{eqnarray*}
and the claim follows.
\end{proof}

The following theorem then follows from Corollary~\ref{cor:W-1logccomp} and Lemma~\ref{lem:W-1logkcomp}.

\begin{theorem}
 $\BAL$ is $\Oh(\lg \min\{k, m\} / \lg\lg \min\{k, m\})$-round-competitive.
\label{teo:W-1logkm}
\end{theorem} 

\subsection{Lower Bounds}
\label{sec:minlb-disjoint}%
In this section we present lower bounds for \Minimum that
hold even for the more restricted case where
the family~$\mathcal{S}$ consists of disjoint sets.

\begin{theorem}
\label{th:lbMcor}%
For arbitrarily large~$m$ and any deterministic algorithm
$\ALG$, there exists an instance with~$m$ sets
and~$k>m$ queries per round,
such that $\opt_k=1$, $\ALG_k \ge W^{-1}(\lg m)$ and $\ALG_k = \Omega(W^{-1}(\lg k))$.
Hence, there is no $\oh(\lg \min \{k, m\} / \lg\lg \min \{k, m\})$-round-competi\-tive deterministic algorithm.
\end{theorem}

\begin{proof}
Fix an arbitrarily large positive integer~$M$.
Consider an instance with $m=M^{M}$ sets, and let $k=M^{M+1}$.
Each set contains~$Mk$ elements, with the $i$-th element having
uncertainty interval $(1+i\varepsilon,100+i\varepsilon)$
for $\varepsilon=1/m$.
The adversary will pick for each set an index~$j$ and set
the $j$-th element to be the minimum, by letting
it have value $1+(j+0.5)\varepsilon$, while the $i$-th
element for $i\neq j$ is given value $100+(i-0.5)\varepsilon$.
The optimal query set for the set is thus its first~$j$ elements.
We assume that an algorithm queries the elements of each set
in order of increasing lower interval endpoints. (Otherwise, the lower bound
only becomes larger.)

Consider the start of
a round when $a \le m$ sets are still active; initially $a = m$.
The adversary observes how the algorithm distributes its~$k$
queries among the active sets and repeatedly adds the active
set with largest number of queries (from the current round)
to a set~$\mathcal{L}$, until the total number of queries from the current
round in sets of~$\mathcal{L}$ is at least $(M - 1) k / M$. Let~$\mathcal{S}'$ denote the
remaining active sets. Note that $|\mathcal{S}'| \ge a / M$. For the sets
in~$\mathcal{L}$, the adversary chooses the minimum in
such a way that a single query in the current round would have
been sufficient to find it, while the sets in~$\mathcal{S}'$ remain
active (and so the optimum must make the same queries in
them that the algorithm has made in the current round, and
these are at most~$k / M$ queries). We continue for~$M$ rounds.
In the $M$-th round, the adversary picks the minimum
in all remaining sets in such way that a single
query in that round would have been sufficient to solve
the set.
The optimal number of queries is then at most
$(M - 1) k / M + M^M = (M - 1) k / M + k / M = k$, and hence
$\opt_k = 1$. On the other hand, we have $\ALG_k = M$.

We can now express this lower bound
in terms of $k$
or $m$ as follows: As $m=M^M$,
we have $\lg m = M \lg M$ and hence $M=W^{-1}(\lg m)$.
As $k=M^{M+1}$, we have $\lg k = (M+1) \lg M$
and hence $M = \Omega(W^{-1}(\lg k))$.
Thus, the theorem follows.
\end{proof}

\begin{theorem}
\label{th:lbkm}%
No deterministic algorithm $\ALG$ attains $\ALG_k \le \opt_k + \oh(\lg \min\{k, m\})$.
\end{theorem}

\begin{proof}
Let $k = m$ be an arbitrarily large integer.
The intervals of the $m$ sets are chosen as in the proof of Theorem~\ref{th:lbMcor},
for a sufficiently large value of~$M$.
Let~$a$ be the number of active sets at the start of a round;
initially $a = m$.
After each round, the adversary considers the set~$S_j$ in
which the algorithm has made the largest number of queries,
which must be at least $k/a$. The adversary picks the minimum element
in~$S_j$ in such a way that a single query in the current
round would have been enough to solve it, and keeps all other
sets active. This continues for~$m$ rounds. The number of
wasted queries is at least $k/m + k/(m-1) + \cdots + k/2 + k - m= k \cdot (H(m)-1)
= k \cdot \Omega(\lg k)$.
As the algorithm must also make all queries in~$\OPT_1$,
the theorem follows from Proposition~\ref{prop:wasted}.
\end{proof}

We conclude thus that the balanced algorithm attains matching upper bounds for disjoint sets.
For non-disjoint sets, a small gap remains between our lower and upper bounds.

\section{Selection}
\label{sec:median}%
An instance of the \Selection problem is given by a set $\mathcal{I}$ of $n$ intervals and
an integer~$i$, $1\le i\le n$.
Throughout this section we denote the $i$-th smallest value in the set of $n$ precise values
by~$v^*$.

\subsection{Finding the Value $v^*$}

If we only want to find the value~$v^*$, then we can adapt the analysis in~\cite{gupta16queryselection} to obtain an algorithm that performs at most $\opt_1 + i - 1$ queries, simply by querying the intervals in the order of their left endpoints.
This is the best possible and can easily be parallelized in $\opt_k + \left\lceil \frac{i-1}{k} \right\rceil$ rounds.
Note that we can assume $i \leq \lceil n / 2 \rceil$, since otherwise we can consider the $i$-th largest value problem, noting that the $i$-th smallest value is the $(n-i+1)$-th largest value.
We also assume that every input interval is either trivial or open, since otherwise (if arbitrary closed intervals are allowed) the problem has a lower bound of~$n$ on the competitive ratio, using the same instance as presented in~\cite{gupta16queryselection} for the problem of identifying the minimum element.

Let $I_{j_1}$ be the interval with the $i$-th smallest left endpoint, and let $I_{j_2}$ be the interval with the $i$-th smallest right endpoint.
Note that any interval~$I_j$ with $u_j < \ell_{j_1}$ or $\ell_j > u_{j_2}$ can be discarded (and the value of $i$ adjusted accordingly).

We analyze the algorithm that simply queries the~$k$ leftmost non-trivial intervals until the problem is solved.

\begin{theorem}
\label{teo:medianleftmost}
For instances of the $i$-th smallest value problem where all input intervals are open or trivial,
there is an algorithm that returns the $i$-th smallest value $v^*$ and uses at most
\begin{displaymath}
 \left\lceil \frac{\opt_1 + i-1}{k} \right\rceil \leq \opt_k + \left\lceil \frac{i-1}{k} \right\rceil
\end{displaymath}
rounds.
\end{theorem}

\begin{proof}
Let~$\mathcal{I}'$ be the set of non-trivial intervals in the input, ordered by non-decreasing left endpoint.
We show that there is a set~$Q \subseteq \mathcal{I}'$ of size at most $\opt_1 + i - 1$ that is a prefix of~$\mathcal{I}'$ in the given ordering and has the property that, after querying~$Q$, the instance is solved.
Given the existence of such a set~$Q$, it is clear that the theorem follows.

Fix an optimum query set~$\OPT_1$, and let~$v^*$ be the $i$-th smallest value.
After querying~$\OPT_1$, assume that there are~$m$ trivial intervals with value~$v^*$.
Note that $m \geq 1$, since it is necessary to determine the value~$v^*$.
Those~$m$ intervals are either queried in~$\OPT_1$ or already were trivial intervals in the input.
We classify the intervals in $\mathcal{I}$ into the following categories, depending on where they fall \emph{after} querying~$\OPT_1$:
 \begin{enumerate}
  \item The set~$M$ (of size~$m$) consisting of trivial intervals whose value is~$v^*$;
  \item The set~$X$ consisting of non-trivial intervals that contain~$v^*$;
  \item The set~$L$ of intervals that are to the left of $v^*$;
  \item The set~$R$ of intervals that are to the right of $v^*$.
 \end{enumerate}
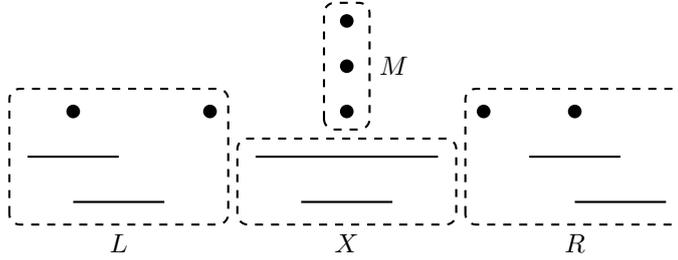
\begin{figure}[ht!]
  \centering
   \begin{tikzpicture}[thick, scale=0.6]
    \draw (0, 1) -- (2, 1);
    \draw (1, 0) -- (3, 0);
    \draw (5, 1) -- (9, 1);
    \draw (6, 0) -- (8, 0);
    \draw (11, 1) -- (13, 1);
    \draw (12, 0) -- (14, 0);
   
    \fill[black] (1, 2) circle (0.15cm);
    \fill[black] (4, 2) circle (0.15cm);
    \fill[black] (7, 2) circle (0.15cm);
    \fill[black] (7, 3) circle (0.15cm);
    \fill[black] (7, 4) circle (0.15cm);
    \fill[black] (10, 2) circle (0.15cm);
    \fill[black] (12, 2) circle (0.15cm);
    
    \draw[rounded corners, dashed] (-0.4, -0.5) rectangle (4.4, 2.5);
    \path (-0.4, -0.5) -- (4.4, -0.5) node[midway, below]{$L$};
    \draw[rounded corners, dashed] (4.6, -0.5) rectangle (9.4, 1.4);
    \path (4.6, -0.5) -- (9.4, -0.5) node[midway, below]{$X$};
    \draw[rounded corners, dashed] (9.6, -0.5) rectangle (14.4, 2.5);
    \path (9.6, -0.5) -- (14.4, -0.5) node[midway, below]{$R$};
    \draw[rounded corners, dashed] (6.5, 1.6) rectangle (7.5, 4.4);
    \path (7.5, 1.6) -- (7.5, 4.4) node[midway, anchor=west]{$M$};
   \end{tikzpicture}
  \caption{An illustration of sets~$M$,~$X$,~$L$ and~$R$ in the proof of Theorem~\ref{teo:medianleftmost}.}
  \label{fig:setsmedian}
\end{figure}
 We illustrate this classification in Figure~\ref{fig:setsmedian}.
 Note that intervals in~$L$ and~$R$ may intersect intervals in~$X$, but cannot contain~$v^*$.
 Let $M^* = M \cap \OPT_1$, $L^* = L \cap \OPT_1$ and $R^* = R \cap \OPT_1$.
 Note that $X \cap \OPT_1 = \emptyset$, and that every interval in $M \setminus M^*$ is trivial in the input.
 
 We claim that the set $Q = (L \cap \mathcal{I'}) \cup X \cup M^* \cup R^*$ is a prefix of~$\mathcal{I}'$ in the given ordering (note that $(X \cup M^* \cup R^*) \setminus \mathcal{I}' = \emptyset$), that querying~$Q$ suffices to solve the instance, and that $|Q| \leq \opt_1 + i - 1$.
 Clearly, every interval in $L \cup X \cup M^*$ comes before all the intervals in $R \setminus R^*$ in the ordering considered.
 It also holds that every interval in~$R^*$ comes before all the intervals in $R \setminus R^*$ in the ordering, since otherwise an interval in~$R^*$ not satisfying this condition could be removed from $\OPT_1$.
 Furthermore, querying all intervals in~$Q$ is enough to solve the instance, because every interval in $R \setminus R^*$ is to the right of~$v^*$, and the optimum solution can decide the problem without querying them.
 Thus it suffices to bound the size of~$Q$.
 Note then that $|L| + |X| \leq i - 1$ since, after querying~$\OPT_1$, the $i$-th smallest interval is in~$M$, and any interval in $L \cup X$ has a left endpoint to the left of~$v^*$.
 Therefore,
 \begin{displaymath}
  |Q| \leq |L| + |X| + |M^*| + |R^*| \leq i - 1 + |M^*| + |R^*| \leq \opt_1 + i - 1,
 \end{displaymath}
which concludes the proof.
\end{proof}

The upper bound of $\left\lceil \frac{\opt_1 + i-1}{k} \right\rceil$ is best possible, because we can construct a lower bound of $\opt_1 + i - 1$ queries to solve the problem.
It uses the same instance as described in~\cite{gupta16queryselection} for the problem of identifying an $i$-th smallest element (but not necessarily finding its precise value).
We include a description of the instance for the sake of completeness.
Consider~$2i$ intervals, comprising~$i$ copies of $(0, 5)$ and~$i$ copies of $\{3\}$.
For the first~$i-1$ intervals $(0, 5)$ queried by the algorithm, the adversary returns a value of~$1$, so the algorithm also needs to query the final interval of the form $(0,5)$ to decide the problem.
Then the adversary sets the value of that interval to~$4$, and querying only that interval would be sufficient for determining that~$3$ is the $i$-th smallest value. Hence any deterministic algorithm makes at least $i$ queries, while $\opt_1=1$.

\subsection{Finding all Elements with Value $v^*$}

Now we focus on the task of finding $v^*$ as well as identifying all intervals in $\mathcal{I}$
whose precise value equals~$v^*$.
For ease of presentation, we assume that all the intervals in $\mathcal{I}$ are closed.
The result can be generalized to arbitrary intervals without any significant new ideas,
but the proofs become longer and require more cases.
A complete proof is included in Appendix~\ref{app:medianopen}.

Let us begin by observing that the optimal query set is easy to characterize.

\begin{lemma}
\label{lem:median_opt}%
Every feasible query set contains all non-trivial intervals that contain~$v^*$.
The optimal query set $\OPT_1$ contains all non-trivial intervals that contain~$v^*$
and no other intervals.
\end{lemma}

\begin{proof}
If a non-trivial interval $I_j$ containing $v^*$ is not queried, one cannot determine
whether the precise value of $I_j$ is equal to $v^*$ or not. Thus, every feasible
query set contains all non-trivial intervals that contain~$v^*$.

Furthermore, it is easy to see that the non-trivial intervals containing $v^*$ constitute
a feasible query set: Once these intervals are queried, one can determine for each interval
whether its precise value is smaller than~$v^*$, equal to~$v^*$, or larger than~$v^*$.
\end{proof}

Let $I_{j_1}$ be the interval with the $i$-th smallest left endpoint, and let $I_{j_2}$ be the interval with the $i$-th smallest right endpoint.
Then it is clear that $v^*$ must lie in the interval $[\ell_{j_1}, u_{j_2}]$, which we call the {\bf target area}.
The following lemma was essentially shown by Kahan~\cite{kahan91queries}; we include a proof for the sake of completeness.

\begin{lemma}[Kahan, 1991]
\label{lem:existsa}%
Assume that the current instance of {\em \Selection} is not yet solved. Then there is at least
one non-trivial interval $I_j$ in $\mathcal{I}$ that contains the target area, i.e., satisfies $\ell_j\le \ell_{j_1}$
and $u_j\ge u_{j_2}$.
\end{lemma}

\begin{proof}
First, assume that the target area is trivial, i.e., $\ell_{j_1}=u_{j_2}=v^*$. If there is no
non-trivial interval in $\mathcal{I}$ that contains $v^*$, then the instance is already solved,
a contradiction.

Now, assume that the target area is non-trivial. Assume that no interval in $\mathcal{I}$
contains the target area. Then all intervals $I_j$ with $\ell_j\le \ell_{j_1}$ have
$u_j < u_{j_2}$. There are at least $i$ such intervals (because
$\ell_{j_1}$ is the $i$-th smallest left endpoint), and hence the $i$-th smallest right endpoint
must be strictly smaller than $u_{j_2}$, a contradiction to the definition of $u_{j_2}$.
\end{proof}

For $k = 1$, there is therefore an online algorithm that makes $\opt_1$ queries: In each
round, it determines the target area of the current instance and queries a non-trivial
interval that contains the target area. (This algorithm was essentially proposed by Kahan~\cite{kahan91queries} for determining all elements with value equal to~$v^*$, without necessarily determining~$v^*$.)
For larger~$k$, the difficulty is how to select additional intervals to query if
there are fewer than $k$ intervals that contain the target area.

The intervals that intersect the target area can be classified into four categories:
\begin{enumerate}[(1)]
 \item $a$ non-trivial intervals $[\ell_j, u_j]$ with $\ell_j \leq \ell_{j_1}$ and $u_j \geq u_{j_2}$; they {\bf contain} the target area;
 \item $b$ intervals $[\ell_j, u_j]$ with $\ell_j > \ell_{j_1}$ and $u_j < u_{j_2}$; they {\bf are strictly contained} in the target area and contain neither endpoint of the target area;
 \item $c$ intervals $[\ell_j, u_j]$ with $\ell_j \leq \ell_{j_1}$ and $u_j < u_{j_2}$; they intersect the target area on the {\bf left};
 \item $d$ intervals $[\ell_j, u_j]$ with $\ell_j > \ell_{j_1}$ and $u_j \geq u_{j_2}$; they intersect the target area on the {\bf right}.
\end{enumerate}
We propose the following algorithm for rounds with $k$ queries:
Each round is filled with as many non-trivial intervals as possible, using the following order: first all intervals of category~(1); then intervals of category~(2); then picking intervals alternatingly from categories~(3) and~(4), starting with category (3). If one of the two categories (3) and (4) is exhausted, the rest of the $k$ queries is chosen from the other category.
Intervals of categories~(3) and~(4) are picked in order of non-increasing length of overlap with the target area, i.e., intervals of category~(3) are chosen in non-increasing order of right endpoint, and intervals of category~(4) in non-decreasing order of left endpoint.
When a round is filled, it is queried, and the algorithm restarts, with a new target area and the intervals redistributed into the categories.

\begin{proposition}
\label{prop:bbound}
At the start of any round, $a \geq 1$
and $b \le a-1$.
\end{proposition}

\begin{proof}
Lemma~\ref{lem:existsa} shows $a\ge 1$. If the target area is trivial,
we have $b=0$ and hence $b\le a-1$. From now on assume that the
target area is non-trivial.

Let $L$ be the set of intervals in $\mathcal{I}$ that lie to the left
of the target area, i.e., intervals $I_j$ with $u_j<\ell_{j_1}$.
Similarly, let $R$ be the set of intervals that lie to the right of the
target area.
Observe that $a+b+c+d+|L|+|R|=n$.

The intervals in $L$ and the intervals of type (1) and (3) include
all intervals with left endpoint at most $\ell_{j_1}$. As $\ell_{j_1}$
is the $i$-th smallest left endpoint, we have $|L|+a+c \ge i$.

Similarly, the intervals in $R$ and the intervals of type (1) and (4)
include all intervals with right endpoint at least $u_{j_2}$.
As $u_{j_2}$ is the $i$-th smallest right endpoint, or equivalently
the $(n-i+1)$-th largest right endpoint, we have
$|R|+a+d\ge n-i+1$.

Adding the two inequalities derived in the previous two
paragraphs, we get
$2a+c+d+|L|+|R|\ge n+1$.
Combined with $a+b+c+d+|L|+|R|=n$, this yields $b \le a-1$.
\end{proof}

\begin{lemma}
\label{lem:lastround}%
If the current round of the algorithm is not the last one,
then the following holds:
If the algorithm queries at least one interval of categories
(3) or~(4), then
the algorithm does not query all intervals of category (3)
that contain~$v^*$, or it does not query all intervals of
category (4) that contain~$v^*$.
\end{lemma}

\begin{proof}
Assume for a contradiction that the algorithm queries at
least one interval of categories~(3) or (4), and that it queries all
intervals of categories (3) and (4) that contain~$v^*$.
Observe that the algorithm also queries all intervals
in categories (1) and (2), as otherwise it would not have
started to query intervals of categories (3) and (4).
Thus, the algorithm has queried all intervals that
contain~$v^*$ and, hence, solved the problem, a contradiction
to the current round not being the last one.
\end{proof}

\begin{theorem}
\label{th:selection}%
There is a $2$-round-competitive algorithm for \Selection.
\end{theorem}

\begin{proof}
Consider any round of the algorithm that is not the last one.
Let~$A$, $B$, $C$ and~$D$ be the sets of intervals of categories (1), (2), (3) and~(4)
that are queried in this round, respectively.
Let~$A^*$, $B^*$, $C^*$ and~$D^*$ be the subsets of $A$, $B$, $C$ and~$D$ that are in~$\OPT_1$, respectively.
By Lemmas~\ref{lem:median_opt} and~\ref{lem:existsa},
$|A| = |A^*|\ge 1$.
Since the algorithm prioritizes category~(1), by Proposition~\ref{prop:bbound} we have $|B| \leq |A| -1$, and thus $|A \cup B| \le 2 \cdot |A| -1 = 2\cdot |A^*| -1
\le 2(|A^*|+|B^*|)-1$.

For bounding the size of $C \cup D$, first note that the order in which the algorithm
selects the elements of categories (3) and (4) ensures that, within each category,
the intervals that contain~$v^*$ are selected first.
By Lemma~\ref{lem:lastround}, there exists a category in which the algorithm
does not query all intervals that contain~$v^*$ in the current round.
If that category is (3), we have $|C|=|C^*|$ and, by the alternating
choice of intervals from (3) and (4) starting with~(3), $|D|\le |C|$
and hence $|C\cup D|\le 2\cdot|C^*|\le 2(|C^*|+|D^*|)$.
If that category is (4), we have $|D|=|D^*|$ and $|C|\le |D|+1$,
giving $|C\cup D|\le 2\cdot|D^*|+1\le 2(|C^*|+|D^*|)+1$.
In both cases, we thus have $|C\cup D|\le 2(|C^*|+|D^*|)+1$.

Combining the bounds obtained in the two previous paragraphs, we
get $|A\cup B\cup C \cup D|\le 2(|A^*|+|B^*|+|C^*|+|D^*|)$.
This shows that, among the queries made in the round, at most half are wasted.
The total number of wasted queries in all rounds except the last
one is hence bounded by $\opt_1$.
Since the algorithm fills each round except possibly the last one and
also queries all intervals in $\OPT_1$, the theorem follows by
Proposition~\ref{prop:wasted}.
\end{proof}

We also have the following lower bound, which proves that our algorithm has the best possible multiplicative factor.
We remark that it uses instances
with $\opt_k=1$, and we do not know how to scale it to larger
values of $\opt_k$. In its present form, it does not
exclude the possibility of an algorithm using at most $\opt_k+1$ rounds.

\begin{lemma}
\label{lem:medianlb}%
There is a family of instances of {\em \Selection} with
$k = i \geq 2$ with $\opt_1\le i$ (and hence $\opt_k=1$) such that any algorithm that
makes $k$ queries in the first round
needs at least two rounds and performs at least $\opt_1 + \lceil (i-1)/2 \rceil$ queries.
\end{lemma}

\begin{proof}
Consider the instance with $i - 1$ copies of interval $[0, 3]$ (called {\bf left-side} intervals), $i - 1$ copies of interval $[5, 8]$ (called {\bf right-side} intervals), and one interval $[2, 6]$ (called {\bf middle} interval).
The precise values are always~$1$ for the left-side intervals, and~$7$ for right-side intervals.
The value of the middle interval depends on the behavior of the algorithm, but in all cases it will be the $i$-th smallest element.
If the algorithm does not query the middle interval in the first round, then we set its value to~$4$, so we have $\opt_1=1$ and the algorithm performs at least $\opt_1 + i = (i + 1) \cdot \opt_1$ queries.
So assume that the algorithm queries the middle interval in the first round.
If it queries more left-side than right-side intervals, then we set the value of the middle interval to~$5.5$, so all right-side intervals must be queried (and all queries of left-side intervals are wasted); otherwise, we set the middle value to~$2.5$.
In either case, we have $\opt_1=i$ and the algorithm wastes at least
$\lceil (i - 1) / 2 \rceil$ queries.
\end{proof}

\section{Relationship with the Parallel Model by Mei{\ss}ner}
\label{sec:meissner}%
In~\cite[Section~4.5]{meissner18querythesis}, Mei{\ss}ner describes a slightly different model for parallelization of queries.
There, one is given a maximum number~$r$ of {\bf batches} that can be performed, and there is no constraint on the number of queries that can be performed in a given batch.
The goal is to minimize the total number of queries performed, and the algorithm is compared to an optimal query set.
The number of uncertain elements in the input is denoted by~$n$.
In this section, we discuss the relationship between this model and the one we described in the previous sections.

Mei{\ss}ner argues that the sorting problem admits a $2$-query-competitive algorithm for $r \geq 2$ batches.
For the minimum problem with one set, she gives an algorithm which is $\lceil n^{1/r} \rceil$-query-competitive,
with a matching lower bound.
She also gives results for the selection and the minimum spanning tree problems.

\begin{theorem}
\label{thm:batchestorounds}
If there is an $\alpha$-query-competitive algorithm that performs at most~$r$ batches, then there is an algorithm that performs at most $\alpha \cdot \opt_k + r - 1$ rounds of~$k$ queries.
Conversely, if a problem has a lower bound of $\beta \cdot \opt_k + t$ on the number of rounds of~$k$ queries, then any algorithm running at most $t + 1$ batches has query-competitive ratio at least~$\beta$.
\end{theorem}

\begin{proof}
Given an $\alpha$-query-competitive algorithm~$A$ on~$r$ batches, we construct an algorithm~$B$ for rounds of~$k$ queries in the following way.
For each batch in~$A$, algorithm~$B$ simply performs all queries in as many rounds as necessary.
In between batches, we may have an incomplete round, but there are only $r-1$ such rounds.
\end{proof}

In view of Mei{\ss}ner's lower bound for the minimum problem with one set mentioned above, the following result is close to being asymptotically optimal for that problem (using $\alpha=1$).

\begin{theorem}
\label{thm:roundstobatches}
If there is an $\alpha$-round-competitive algorithm for rounds of~$k$ queries, with~$\alpha$ independent of~$k$, then there is an algorithm that performs at most~$r$ batches with query-competitive ratio $\Oh(\alpha \cdot n^{\lfloor \alpha \rfloor / (r - 1)})$, with $r \geq \lfloor \alpha \rfloor \cdot x + 1$ for an arbitrary natural number~$x$.
In particular, for $r \geq \lfloor \alpha \rfloor \cdot \lg n + 1$, the query-competitive factor is~$\Oh(\alpha)$.
\end{theorem}

\begin{proof}
Assume $r = \lfloor \alpha \rfloor \cdot x + 1$ for some natural number~$x$; otherwise we can simply leave some batches unused.
Given an $\alpha$-round-competitive algorithm~$A$ for rounds of~$k$ queries, we construct an algorithm~$B$ that performs at most~$r$ batches.
We group them into sequences of~$\lfloor \alpha \rfloor$ batches.
For the $i$-th sequence, for $i = 1, \ldots, x$, algorithm~$B$ runs algorithm~$A$ for~$\lfloor \alpha \rfloor$ rounds with $k = n^{(i-1)/x}$, until the problem is solved.
If the problem is not solved after $\lfloor \alpha \rfloor \cdot x$ batches, then algorithm~$B$ queries all the remaining intervals in one final batch.

To determine the query-competitive ratio, consider the number~$i$ of sequences of~$\lfloor \alpha \rfloor$ batches the algorithm executes.
If the problem is solved during the $i$-th sequence, then algorithm~$B$ performs at most
$
\lfloor \alpha \rfloor \cdot (\sum_{j=0}^{i-1} n^{j/x})
= \lfloor \alpha \rfloor \cdot \Theta(n^{(i-1)/x})
$
queries.
(If the problem is solved during the last batch, it performs at most $n \leq \lfloor \alpha \rfloor \cdot n^{x/x}$ queries.)
On the other hand, we claim that, if the problem is not solved after the $(i-1)$-th sequence, then the optimum solution queries at least~$n^{(i-2)/x}$ intervals.
This is because algorithm~$A$ is $\alpha$-round-competitive, so whenever the algorithm performs a sequence of~$\lfloor \alpha \rfloor$ rounds for a certain value of~$k$ and does not solve the problem, it follows that the optimum solution requires more than one round for this value of $k$, and hence more than $k$ queries.
Thus, the query-competitive ratio is at most $\lfloor \alpha \rfloor \cdot \Theta(n^{1/x}) = \Theta(\alpha \cdot n^{\lfloor \alpha \rfloor/(r-1)})$.
\end{proof}

Therefore, an algorithm that uses a constant number of batches implies an algorithm with the same asymptotic round-competitive ratio for rounds of~$k$ queries.
On the other hand, some problems have worse query-competitive ratio if we require few batches, even if we have round-competitive algorithms for rounds of~$k$ queries, but the ratio is preserved by a constant if the number of batches is sufficiently large.

\section{Final Remarks}
\label{sec:future}

We propose a model with parallel queries and the goal of minimizing the number of query rounds when solving uncertainty problems.
Our results show that, even though the techniques developed for the sequential setting can be utilized in the new framework, they are not enough, and some problems are harder (have a higher lower bound on the competitive ratio).

One interesting open question is how to extend our algorithms for \Minimum to the variant where it is not necessary to return the precise minimum value, but just to identify the minimum element.
Another problem one could attack is the following generalization of \Selection: Given multiple sets $S_1, \ldots, S_m \subseteq \mathcal{I}$ and indices $i_1, \ldots, i_m$, identify the $i_j$-smallest precise value and all elements with that value in~$S_j$, for $j = 1, \ldots, m$.
It would be interesting to see if the techniques we developed for \Minimum with multiple sets can be adapted to \Selection with multiple sets.

It would be nice to close the gaps in the round-competitive ratio, to understand if the analysis of Algorithm~\ref{fig:algA} is tight, and
to study whether randomization can help to obtain better upper bounds.
One could also study other problems in the parallel model, such as the minimum spanning tree
problem.

\bibliographystyle{plainurl}
\bibliography{../../queries}

\appendix

\section{Asymptotic Growth of $W^{-1}$}
\label{app:W-1}

For the sake of completeness, we include a proof of the following:

\begin{proposition}
\label{prop:w-1theta}
Let $W(x) = x \lg x$. It holds that $W^{-1}(x) = \Theta(x / \lg x)$.
\end{proposition}

\begin{proof}
First we claim that $W^{-1}(x) \geq 2$ for $x \geq 2$.
It holds that $W^{-1}(x)$ is injective for $x > 0$, because $y \lg y$ is increasing for $y > 1$, and $y \lg y \leq 0$ for $0 < y \leq 1$ but $y \lg y > 0$ for $y > 1$.
Thus $W^{-1}(2)$ is unique and it is easy to check that $W^{-1}(2) = 2$.
By implicit differentiation,
\begin{eqnarray*}
x & = & W^{-1}(x) \lg W^{-1}(x) \\
1 & = & \left( \lg W^{-1}(x) + \frac{1}{\ln 2} \right) \frac{dW^{-1}(x)}{dx} \\
\frac{dW^{-1}(x)}{dx} & = & \frac{1}{\left( \lg W^{-1}(x) + \frac{1}{\ln 2} \right)},
\end{eqnarray*}
which is greater than zero for $x > 0$, because $W^{-1}(x) > 1$ for $x > 0$, because $\lg y$ is not defined for $y \leq 0$, and $\lg y \leq 0$ for $0 < y \leq 1$.
Therefore $W^{-1}(x)$ is increasing and $W^{-1}(x) \geq 2$ for $x \geq 2$.

Now we prove the asymptotic bounds for $x \geq 2$.
First,
\begin{displaymath}
\frac{x}{\lg x} = \frac{W^{-1}(x) \lg W^{-1}(x)}{\lg (W^{-1}(x) \lg W^{-1}(x))} \leq W^{-1}(x) \cdot \frac{\lg W^{-1}(x)}{\lg W^{-1}(x)} = W^{-1}(x),
\end{displaymath}
where the inequality holds because $y \lg y \geq y$ for $y = W^{-1}(x) \geq 2$.
Furthermore,
\begin{displaymath}
\frac{2x}{\lg x} = \frac{2 W^{-1}(x) \lg W^{-1}(x)}{\lg (W^{-1}(x) \lg W^{-1}(x))} \geq \frac{2 W^{-1}(x) \lg W^{-1}(x)}{\lg (W^{-1}(x))^2} = \frac{2 W^{-1}(x) \lg W^{-1}(x)}{2 \lg W^{-1}(x)} = W^{-1}(x),
\end{displaymath}
where the inequality holds because $y \lg y \leq y^2$ for $y = W^{-1}(x) \geq 2$.

Thus, we have shown that $\frac{x}{\lg x} \le W^{-1}(x) \le \frac{2x}{\lg x}$ for all $x\ge 2$.
\end{proof}

\section{Selection with Arbitrary Intervals}
\label{app:medianopen}

\newcommand{\ta}{\mathrm{ta}}

In this section we prove that Theorem~\ref{th:selection} also holds for arbitrary intervals.

An instance of the \Selection problem is given by a set $\mathcal{I}$ of $n$ intervals and
an integer~$i$, $1\le i\le n$.
The $i$-th smallest value in the set of $n$ precise values is denoted by~$v^*$.
The task is to find $v^*$ as well as identify all intervals in $\mathcal{I}$
whose precise value equals~$v^*$.

We allow arbitrary intervals as input: trivial intervals containing a single value,
open intervals, closed intervals, and intervals that are closed on one side and
open on the other.

Lemma~\ref{lem:median_opt} and its proof hold also for arbitrary intervals without
any changes.

We call the left endpoint $\ell_i$ of an interval $I_i$ an \textbf{open left endpoint} if the interval
does not contain $\ell_i$ and a \textbf{closed left endpoint} otherwise.
The definitions of the terms \textbf{open right endpoint} and \textbf{closed right endpoint}
are analogous. When we order the left endpoints of the intervals in non-decreasing
order, equal left endpoints are ordered as follows: closed left endpoints come before
open left endpoints. When we order the right endpoints of the intervals in non-decreasing
order, equal right endpoints are ordered as follows: open right endpoints come before
closed right endpoints. Equal left endpoints that are all open can be ordered arbitrarily,
and the same holds for equal left endpoints that are all closed, for equal right endpoints
that are all open, and for equal right endpoints that are all closed.
Informally, the order of left endpoints orders intervals in
order of the ``smallest'' values they contain, and the order of right endpoints
orders intervals in order of the ``largest'' values they contain.
We call the resulting order of left endpoints $\preceq_L$ and the resulting
order of right endpoints $\preceq_U$.
We say that a right endpoint $u_{i_1}$ \textbf{strictly precedes} a right endpoint
$u_{i_2}$ if either $u_{i_1}<u_{i_2}$ or $u_{i_1}=u_{i_2}$ and $u_{i_1}$ is an
open right endpoint and $u_{i_2}$ is a closed right endpoint.

Let $I_{j_1}$ be the interval with the $i$-th smallest left endpoint (i.e., the $i$-th left endpoint
in the order~$\preceq_L$), and let $I_{j_2}$ be the interval with the $i$-th smallest right
endpoint (i.e, the $i$-th right endpoint in the order $\preceq_U$).
Then it is clear that $v^*$ must lie in the interval $I_{\ta}$, which we call the {\bf target area} and define
as follows:
\begin{itemize}
\item If $\ell_{j_1}$ is an open left endpoint of $I_{j_1}$ and
         $u_{j_2}$ is an open right endpoint of $I_{j_2}$,
	 then $I_{\ta}=(\ell_{j_1},u_{j_2})$.
\item If $\ell_{j_1}$ is an open left endpoint of $I_{j_1}$ and
         $u_{j_2}$ is a closed right endpoint of $I_{j_2}$,
	 then $I_{\ta}=(\ell_{j_1},u_{j_2}]$.
\item If $\ell_{j_1}$ is a closed left endpoint of $I_{j_1}$ and
         $u_{j_2}$ is an open right endpoint of $I_{j_2}$,
	 then $I_{\ta}=[\ell_{j_1},u_{j_2})$.
\item If $\ell_{j_1}$ is a closed left endpoint of $I_{j_1}$ and
         $u_{j_2}$ is a closed right endpoint of $I_{j_2}$,
	 then $I_{\ta}=[\ell_{j_1},u_{j_2}]$.
\end{itemize}

%
The following lemma was essentially shown by Kahan~\cite{kahan91queries}; for the sake
of completeness, we give a proof for arbitrary intervals.

\begin{lemma}[Kahan, 1991; version of Lemma~\ref{lem:existsa} for arbitrary intervals]
\label{lem:existsa-arb}%
Assume that the current instance of {\em \Selection} is not yet solved. Then there is at least
one non-trivial interval $I_j$ in $\mathcal{I}$ that contains the target area~$I_\ta$.
\end{lemma}

\begin{proof}
First, assume that the target area is trivial, i.e., $I_\ta=\{v^*\}$. If there is no
non-trivial interval in $\mathcal{I}$ that contains $v^*$, then the instance is already solved,
a contradiction.

Now, assume that the target area $I_\ta$ is non-trivial. Assume that no interval in $\mathcal{I}$
contains the target area. Then all intervals $I_j$ whose left endpoint is not after $\ell_{j_1}$
in the order of left endpoints must have a right endpoint that strictly precedes $u_{j_2}$.
There are at least $i$ such intervals (because
$\ell_{j_1}$ is the $i$-th smallest left endpoint), and hence the $i$-th smallest right endpoint
must strictly precede $u_{j_2}$ in the order of right endpoints, a contradiction to the
definition of $u_{j_2}$.
\end{proof}

The intervals that intersect the target area can be classified into four categories:
\begin{enumerate}[(1)]
 \item $a$ non-trivial intervals that {\bf contain} the target area;
 \item $b$ intervals that {\bf are strictly contained} in the target area such that the target area contains at least one point to the left of the interval and at least one point to the right of the interval.
 \item $c$ intervals that contain some part of $I_{\ta}$ at the left end and do
 not contain some part of $I_{\ta}$ on the right end.
 Formally, an interval
 $I_i$ with closed right endpoint $u_i$ is in this category if $u_i\in I_\ta$,
 $I_i\cap I_\ta=\{v\in I_{\ta}\mid v\le u_i\}$ and
 $\{v\in I_{\ta}\mid v> u_i\}\neq\emptyset$.
 Moreover, an interval
 $I_i$ with open right endpoint $u_i$ is in this category if $u_i\in I_\ta$ and
 $I_i\cap I_\ta=\{v\in I_{\ta}\mid v < u_i\}\neq \emptyset$.
 \item $d$ intervals that contain some part of $I_{\ta}$ at the right end and do
 not contain some part of $I_{\ta}$ on the left end.
 Formally, an interval
 $I_i$ with closed left endpoint $\ell_i$ is in this category if $\ell_i\in I_\ta$,
 $I_i\cap I_\ta=\{v\in I_{\ta}\mid v\ge \ell_i\}$ and
 $\{v\in I_{\ta}\mid v< \ell_i\}\neq\emptyset$.
 Moreover, an interval
 $I_i$ with open left endpoint $\ell_i$ is in this category if $\ell_i\in I_\ta$ and
 $I_i\cap I_\ta=\{v\in I_{\ta}\mid v > \ell_i\}\neq \emptyset$.
\end{enumerate}

We propose the following algorithm for rounds with $k$ queries.
Each round is filled with as many intervals as possible, using the following order: First all intervals of category~(1); then intervals of category~(2); then picking intervals alternatingly from categories~(3) and~(4), starting with category (3). If one of the two categories is exhausted, the rest of the $k$ queries is chosen from the other category.
Intervals of categories~(3) and~(4) are picked in order of non-increasing length of overlap with the target area.
More precisely, intervals of category~(3) are chosen according to the reverse of the order~$\preceq_U$ of their right endpoints,
and intervals of category~(4) are chosen according to the order~$\preceq_L$ of their left endpoints.
When a round is filled, it is queried, and the algorithm restarts, calculating a new target area and redistributing the intervals into the categories.

\begin{proposition}[Version of Proposition~\ref{prop:bbound} for arbitrary intervals]
\label{prop:bbound-arb}
At the start of any round, $a \geq 1$
and $b \le a-1$.
\end{proposition}

\begin{proof}
Lemma~\ref{lem:existsa-arb} shows $a\ge 1$. If the target area is trivial,
we have $b=0$ and hence $b\le a-1$. From now on assume that the
target area is non-trivial.

Let $L$ be the set of intervals in $\mathcal{I}$ that lie to the left
of $I_\ta$ (and have empty intersection with $I_\ta$).
Similarly, let $R$ be the set of intervals that lie to the right of
$I_\ta$ (and have empty intersection with $I_\ta$).
Observe that $a+b+c+d+|L|+|R|=n$.

The intervals in $L$ and the intervals of type (1) and (3) include
all intervals with left endpoint not after $\ell_{j_1}$ in the order~$\preceq_L$.
As $\ell_{j_1}$ is the $i$-th left endpoint in that order, we have $|L|+a+c \ge i$.

Similarly, the intervals in $R$ and the intervals of type (1) and (4)
include all intervals with right endpoint not before $u_{j_2}$ in the order~$\preceq_U$.
As $u_{j_2}$ is the $i$-th smallest right endpoint in that order, or equivalently
the $(n-i+1)$-th largest right endpoint in that order, we have
$|R|+a+d\ge n-i+1$.

Adding the two inequalities derived in the previous two
paragraphs, we get
$2a+c+d+|L|+|R|\ge n+1$.
Combined with $a+b+c+d+|L|+|R|=n$, this yields $b \le a-1$.
\end{proof}

Lemma~\ref{lem:lastround} and its proof hold for arbitrary
intervals without any changes.

\begin{theorem}
There is a $2$-round-competitive algorithm for \Selection\ even
if arbitrary intervals are allowed as input.
\end{theorem}

\begin{proof}
The proof is identical to the proof of Theorem~\ref{th:selection},
except that Lemma~\ref{lem:existsa-arb} and Proposition~\ref{prop:bbound-arb}
are used in place of Lemma~\ref{lem:existsa} and Proposition~\ref{prop:bbound},
respectively.
\end{proof}

\end{document}

\endinput